%% file: SODA.tex
\documentclass{article}
\usepackage[ruled]{algorithm2e} 
\usepackage[utf8]{inputenc}
\usepackage{mathtools}
\usepackage{comment}
\usepackage{amsmath}
\usepackage{algpseudocode}
\usepackage{amsfonts,amssymb,amsthm,boxedminipage,color,url,fullpage}
\usepackage{amsfonts}
\usepackage{amssymb}
\usepackage{amsthm}
\usepackage{boxedminipage}
\usepackage{color}
\usepackage{url}
\usepackage{fullpage}
\usepackage{mathtools}
\usepackage[numbers]{natbib}

\usepackage{thmtools} 
\usepackage{thm-restate}

\usepackage{enumitem}
\usepackage{tcolorbox}
\usepackage[colorlinks=true,citecolor=blue,linkcolor=blue,urlcolor=blue]{hyperref}

\usepackage[labelfont=bf]{caption}
\usepackage{aliascnt,cleveref}
\usepackage{authblk}
\usepackage{accents}
\usepackage{tikz}
\usetikzlibrary{positioning,arrows}
\usepackage{pgfplots}
\pgfplotsset{width=10cm,compat=1.9}
\usepgfplotslibrary{external}
\allowdisplaybreaks

\usepackage[]{color-edits}

\input{macros}

\newcommand{\newShare}{strong EEFX share}

\newcommand{\bv}{\bar{v}}

\theoremstyle{plain}
\newtheorem{theorem}{Theorem}[section]
\newtheorem{lemma}[theorem]{Lemma}
\newtheorem{claim}[theorem]{Claim}
\newtheorem{corollary}[theorem]{Corollary}

\newtheorem{proposition}[theorem]{Proposition}

\newtheorem{observation}[theorem]{Observation}

\newtheorem{example}[theorem]{Example}

\theoremstyle{definition}
\newtheorem{definition}[theorem]{Definition}

\theoremstyle{definition}

\theoremstyle{definition}

\allowdisplaybreaks

\title{The Power of Share-Based Notions in Proving \\ Envy-Based Fairness Guarantees}

\author[1,2]{Hannaneh Akrami}
\author[3]{Uriel Feige}
\author[4]{Ryoga Mahara}
\author[1]{Kurt Mehlhorn}
\author[5]{Nidhi Rathi}

\affil[1]{Max Planck Institute Informatics and Universität des Saarlandes, Germany}
\affil[2]{Hertz Chair for Algorithms and Optimization, University of Bonn, Germany}
\affil[3]{Weizmann Institute, Israel}
\affil[4]{The University of Tokyo, Japan}
\affil[5]{University of Warsaw, Poland}

\date{}

\begin{document}

\maketitle
\begingroup
\renewcommand{\thefootnote}{}
\footnotetext{$^*$This version merges our earlier arXiv preprint with the results of \cite{F25}, and includes substantial revisions throughout.}
\endgroup

\begin{abstract}

    We study the fundamental problem of fairly allocating a set of indivisible goods among agents with monotone valuations. We introduce a new share-based fairness notion, the \emph{residual maximin share} (RMMS), and show that it provides a common framework for several existing ``lone-divider'' style techniques in fair division. RMMS enjoys two key properties: it is \emph{feasible} and \emph{self-maximizing}.  
    Using RMMS, we obtain simple proofs of the existence of partial allocations that are both RMMS and envy-free up to any good (EFX), as well as complete allocations that are both RMMS and \emph{envy-free up to one good} (EF1) (even the stronger notion of EFL), thereby unifying and strengthening several previously known results.

    We further demonstrate the power of the share-based approach by studying the compatibility of fairness notions related to the long-standing EFX problem. While allocations satisfying either \emph{epistemic EFX} (EEFX) or \emph{envy-freeness up to one good} (EF1) are known to exist for general monotone valuations, whether they can always be achieved simultaneously has remained open in all settings where the existence of EFX itself is unresolved. For \emph{additive} valuations, we resolve this question positively by proving the existence of allocations that are both EEFX and, in fact, satisfy the stronger notion of EFL. Our approach introduces another share-based fairness notion, the \emph{strong EEFX share}, which implies EEFX feasibility of bundles. We show that the strong EEFX share is upper bounded by RMMS, allowing us to leverage the RMMS framework to establish the existence of EEFX+EFL allocations. This answers the main open question posed by \cite{EpistemicMonotone}.

    Finally, although our algorithm for computing EEFX and EF1 allocations may take exponential time in general, we develop a polynomial-time algorithm for the case of \emph{restricted} additive valuations. This algorithm circumvents the lone-divider approach entirely, and instead exploits the structural properties specific to restricted additive valuations to achieve both EF1 and EEFX guarantees. 

\end{abstract}



\input{new-intro}

\section{Preliminaries}\label{sec:pre}
A \emph{fair division instance} is given by $\instance$, where $N=[n]$, and for any agent $i\in N$, $v_i: 2^M \rightarrow \mathbb{R}_{\ge 0}$ denotes their valuation function over the set of items. 
For all $i \in N$, we assume $v_i$ is \emph{normalized},
i.e., $v_i(\emptyset)=0$, and \emph{monotone}, i.e., for all $i \in N$, $g \in M$ and $S \subseteq M$,
$v_i(S \cup \{g\}) \ge v_i(S)$. 
Moreover, we say a valuation function $v$ is \emph{additive} if for all subsets $S\subseteq M$, we have $v(S) = \sum_{g \in S} v(\{g\})$.

An \emph{allocation} $X=(X_1,X_2,\ldots,X_n)$ of the items among agents is a partition of items into $n$ bundles
such that bundle $X_i$ is allocated to agent $i$. That is, we have $X_i \cap X_j = \emptyset$ for all
$i,j \in N$ with $i \neq j$ and $\bigcup_{i \in [n]} X_i = M$. In case $\bigcup_{i \in [n]} X_i \subset M$, we call the allocation as \emph{partial} and denote the set of unallocated items by $X_0$. Whenever not specified, the allocation is \emph{complete}; i.e., $\bigcup_{i \in [n]} X_i = M$.

\begin{definition}[EFX (Envy-freeness up to any good)~\cite{caragiannis2019unreasonable}]\label{def:efx}
    An allocation $X$ is \emph{envy-free up to any good (EFX)} if for all agents $i,j\in N$,
    for any item $g\in X_j$, we have $v_i(X_i)\ \ge\ v_i(X_j\setminus\{g\}).$
\end{definition}

\begin{definition}[EFX Partition]
    A partition $(P_1, \ldots, P_k)$ of $M$ into $k$ bundles is EFX for agent $i$ if $v_i(P_j) \geq v_i(P_{j'} \setminus \{g\})$ for all $j,j' \in [k]$ and $g \in P_{j'}$. 
\end{definition}
Plaut and Roughgarden~\cite{plaut2020almost} proved that when agents have identical monotone valuations, an EFX allocation exists. 
Equivalently, for any agent $i$, there exists an EFX partition of $M$ into $n$ bundles.
\begin{theorem}[\cite{plaut2020almost}]
    For an agent $i$ with monotone valuation, there exists an EFX partition of $M$ into $n$ bundles for agent $i$. 
\end{theorem}

\begin{definition}[EF1 (Envy-freeness up to one good)~\cite{budish2011combinatorial}]
    An allocation $X$ is \emph{envy-free up to one good (EF1)} if for all agents $i,j\in N$,
    either $v_i(X_i)\ge v_i(X_j)$, or there exists an item $g\in X_j$ such that $v_i(X_i)\ \ge\ v_i(X_j\setminus\{g\}).$
\end{definition}

\emph{Envy-freeness up to one less-preferred good (EFL)} \cite{barman2018groupwise} is strictly more demanding fairness notion than EF1. 

\begin{definition}[EFL (Envy-freeness up to one less-preferred good)~\cite{barman2018groupwise}]
    An allocation $X$ is \emph{envy-free up to one less-preferred good (EFL)} if for all agents $i,j\in N$,
    at least one of the following conditions holds:
    \begin{enumerate}
    \item $|\{g\in X_j \mid v_i(\{g\})>0\}|\le 1$, 
    \item there exists an item $g\in X_j$ such that
    $v_i(X_i) \ge v_i(X_j\setminus\{g\})$ and
    $v_i(X_i)\ \ge\ v_i(\{g\})$.
    \end{enumerate}
\end{definition}

\begin{definition}[EEFX feasibility]
    A bundle $B \subseteq M$ is \emph{EEFX-feasible} for an agent $i$ iff there exists a partition of $M \setminus B$ into $n-1$ bundles $P_1, \ldots, P_{n-1}$ such that $v_i(B) \geq v_i(P_j \setminus \{g\})$ for all $j \in [n-1]$ and $g \in P_j$. If $B$ is EEFX-feasible for agent $i$, such a partition $P=(P_1, \ldots, P_{n-1})$ is called as an \emph{EEFX certificate} for $B$. Similarly, $B$ is \emph{EEFX-infeasible} for agent $i$ iff no such partition exists. 
\end{definition}

\begin{definition}[EEFX (Epistemic envy-freeness up to any good)~\cite{Caragiannis2023}]
An allocation $X$ is \emph{epistemic EFX (EEFX)} if for every agent $i\in N$, the bundle $X_i$ is EEFX-feasible for $i$.   
\end{definition}

\begin{definition}[Maximin Share (MMS)]
    Let $\Pi$ be the set of all partitions of $M$ into $n$ bundles. Then, $$\text{MMS}_i = \max_{P \in \Pi} \min_{j \in [n]} v_i(P_j).$$    
\end{definition}

For $\alpha \geq 0$, an allocation is called $\alpha$-MMS if every agent values their bundle at least $\alpha$ times their MMS value. When $\alpha = 1$, the allocation is simply called an MMS allocation. 

Another share-based notion that is highly related to the notion of EEFX is the \emph{minimum EFX share} (MXS) introduced by Caragiannis et al. \cite{Caragiannis2023}. 
\begin{definition}[MXS (Minimum EFX share) \cite{Caragiannis2023}]
    For any agent $i \in N$, the minimum EFX share of $i$ is 
    \[  \text{MXS}_i = \min\{v_i(S) \mid S \text{ is EEFX-feasible for $i$} \}. \]
\end{definition}
An allocation is said to be MXS if all agents value their bundle at least at their MXS value. Clearly, the existence of EEFX allocations implies the existence of MXS allocations.

For a share $s$ (such as MMS, MXS), an \emph{$s$-allocation} (\emph{partial $s$-allocation}, respectively) is an allocation (partial allocation, respectively) in which every agent gets a bundle that she values at least as high as her $s$ share value. 
\begin{definition}[Feasiblity]
    A share $s$ is \emph{feasible} for a class $C$ of valuations if in every allocation instance with valuations from the class $C$, there is an $s$-allocation. For a given ratio $0 < \rho < 1$ and a share $s$, $\rho$-$s$ is a share whose value is $\rho$ times the value of the $s$ share.    
\end{definition}
For a given share $s$ and valuation $v$, a partition $(P_1, \ldots , P_n)$ of $M$ will be referred to as a $v$-acceptable $n$-partition with respect to $s$ if $v(P_j) \geq s(M, v, n)$ for every $j$. We omit $s$ when it is clear by context.

\begin{definition}[Self-feasibility]
    A share $s$ is \emph{self-feasible} for class $C$, if every $v \in C$ has a $v$-acceptable $n$-partition.  
\end{definition}
By definition, feasibility implies self-feasibility. We now introduce the notion of \emph{residual self-feasibility} in order to define the notion of \emph{residual maximin share (RMMS)}.
\begin{definition}[Residual Self-feasibility]
    A share $s$ is \emph{residual self-feasible} for class $C$, if for every $v \in C$, every $0 \leq k < n$, and every $k$ bundles each of value strictly less than $s(M, v_i, n)$, after removing these bundles, there is a $v$-acceptable $(n-k)$-partition (of the set of remaining items).
\end{definition}
The MMS is the self-feasible share with the highest possible value. We now introduce
an analogous notion of \emph{RMMS share}, which is the residual self-feasible share with the highest possible value.

\begin{definition}[Residual Maximin Share (RMMS)]\label{def:rmms}
    The value of the residual maximin share (RMMS) for valuation $v$, denoted as RMMS$(M, v, n)$, is the highest value $t$ that is residual self-feasible. That is, for every $0 \leq k < n$, after removing $k$ bundles each of value (under $v$) strictly less than $t$, there is an $(n-k)$ partition of the set of remaining items, where each part has value (under $v$) at least $t$. 
\end{definition} 

By definition, the RMMS is never larger than the MMS. By design, the RMMS is at least as large as the MXS. Figure~\ref{fig:0} illustrates the relationships among the fairness notions introduced so far for additive valuations.

\begin{definition}\label{def:natural}
    We say that valuation $v_i$ dominates valuation $v_j$ if for every bundle $S$, $v_i(S) \ge v_j(S)$. For a given $\varepsilon > 0$, we say that two valuations are $\varepsilon$-close if for every bundle $S$, $|v_i(S) - v_j(S)| \le \varepsilon$. A share is {\em monotone} if whenever $v_i$ dominates $v_j$, the share value for $v_i$ is at least as high as that for $v_j$. A share is {\em 1-Lipschitz} if whenever $v_i$ and $v_j$ are $\varepsilon$-close, their share value differs by at most $\varepsilon$.
\end{definition}

As an example, MMS is both monotone and 1-Lipschitz, whereas $\frac{2}{3}$-MMS is monotone but not 1-Lipschitz (and not even continuous, see~\cite{BF25}). The following example shows that MXS is not monotone

\begin{example}\label{example:mxs}
Consider an instance with $n=2$ agents and five items with additive values $(4,3,3,3,1)$.
The MXS value is $7$. Indeed, every EEFX-feasible bundle has value at least $7$, while the bundle consisting of the items valued $4$ and $3$ has value exactly $7$ and is EEFX-feasible.

Now increase the value of the first item from $4$ to $5$, obtaining the valuation $(5,3,3,3,1)$.
The MXS value drops to $6$: the bundle consisting of the items valued $5$ and $1$ has value $6$ and is EEFX-feasible, while no EEFX-feasible bundle has smaller value.

Thus, MXS is not monotone: increasing the value of an item can decrease the agent's MXS value.
\end{example}

Being {\em self-maximizing} is a property of shares that implies both monotonicity and 1-Lipschitz (see~\cite{BF25}), and also relates to incentives for agents to report their true valuations. (If an agent is risk averse and wishes to maximizes the minimum possible value that she might get by an allocation mechanism that only ensures that she gets at least her share value, she has no incentive to report a valuation $v'$ instead of her true valuation $v$.)

\begin{definition}
    \label{def:selfMaximizing}
    We say that a share $s$ is \emph{self-maximizing} if for every two valuations $v$ and $v'$, there is a bundle $T$ satisfying $v'(T) \ge s(\items, v', n)$ (hence, $T$ is feasible for $v'$ under share $s$) such that $v(T) \le s(\items, v, n)$ (hence, even the worst bundle that is feasible under the share $v$ is at least as valuable as $T$, according to the true valuation $v$ of the agent).
\end{definition}

The MMS is a self-maximizing share, but is not feasible, not even for additive valuations~\cite{KPW18}.  For the special case of additive valuations, the {\em nested share} of~\cite{BF25} is both feasible and self maximizing, and its value is at least $\frac{2n}{3n-1}$-MMS. (Moreover, it can be computed in polynomial time, and so can a feasible allocation.) 

Lipton, Markakis, Mossel and Saberi~\cite{lipton2004approximately} designed an allocation algorithm that produces EF1 allocations. We refer to their algorithm as the LMMS algorithm. The LMMS algorithm starts from the empty allocation, and then allocates items one by one. However, as observed in multiple previous works, it can also be started from any partial allocation. The following lemma (whose proof is omitted, as it follows trivially from~\cite{lipton2004approximately})  summarizes its properties in this case.

\begin{lemma}
    \label{lem:LMMS}
    Let ${\mathcal P} = (P_0, P_1, \ldots, P_n)$ be a partial allocation, and let ${\mathcal A} = (A_1, \ldots, A_n)$ be the final allocation if the LMMS algorithm is executed starting at $P$. Then ${\mathcal A}$ satisfies the following properties:
    \begin{enumerate}
        \item There is a matching $\pi$ between the bundles $A_1, \ldots, A_n$ and $P_1, \ldots, P_n$ such that $P_{\pi(i)} \subset A_i$ for every $i$.
        \item For every agent $a_i$, $v_i(A_i) \ge v_i(P_i)$.
        \item For any $j$, if $A_j \not= P_{\pi(j)}$, let $g_j$ be the last item inserted into $A_j$ by the LMMS algorithm. Then for every agent $i$, $v_i(A_i) \ge v_i(A_j \setminus g_j)$.
    \end{enumerate}
\end{lemma}
The \emph{strong EEFX share} is a key new concept introduced in this work, that we define next.

\begin{definition}[Strong EEFX Share]\label{def:stronfEEFXShare}
    For any agent $i \in N$, the \emph{\newShare}~of $i$ is 
    \[  \theta_i = \min\{v_i(S) \mid \text{ any set $T$ with $v_i(T) \ge v_i(S)$ is EEFX-feasible for $i$} \}. \]
\end{definition}

Note that since $M$ is EEFX-feasible for any agent $i$, the \newShare~is well-defined. Moreover, it is easy to observe that $\theta_i \geq \text{MXS}_i$ for any agent $i$. See \Cref{fig:EEFX-share} for better intuition.

\begin{figure}[t]
    \centering
    \includegraphics[width=1\linewidth]{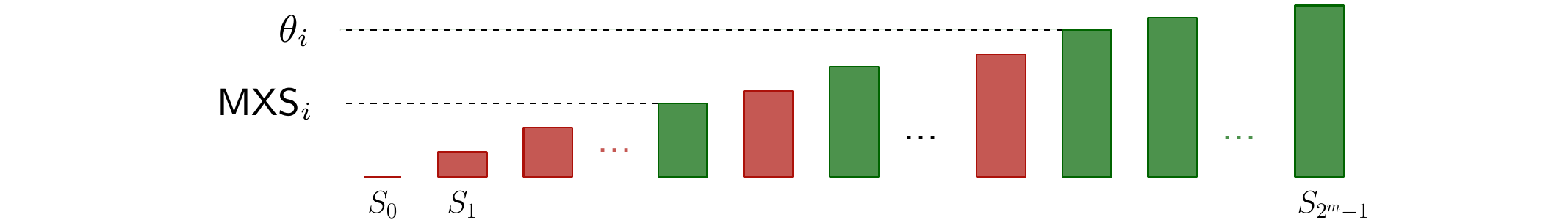}
    \caption{Visualization of the MXS and the strong EEFX share. All subsets of $M$ are sorted in increasing value for agent $i$. Red bundles are EEFX-infeasible, while green bundles are EEFX-feasible.}
    \label{fig:EEFX-share}
\end{figure}

We present an example to highlight the fact that EEFX feasibility is not monotone in bundle-value, and hence, motivating the definition of the strong EEFX share. We further show that $\theta_i$ can be strictly higher than $\text{MXS}_i$, making the former a stronger fairness guarantee than the latter.

\begin{example}\label{example}
Consider a simple instance $\instance$ with $N = [2]$, $M = \{g_1, g_2, \ldots, g_{6}\}$ where, for agent $1$, $g_1$ and $g_2$ are large items of value $4$ and the remaining ones are small items of value $1$. That is, $v_1(g_1)=v_1(g_2)=4$ and $v_1(g_3)=
\dots  = v_1(g_6)=1$. Note that the bundle of value 4 containing only the small items is EEFX-feasible, but the  bundle of value 5 containing one large and one small item is not EEFX-feasible. It is easy to see that $\text{MXS}_1 = 4$ while $\theta_1 = 6$.                      
\end{example}

\begin{definition}\label{def:stronglyEEFX}
    A set $S \subseteq M$ is a \emph{strongly EEFX-feasible} bundle for agent $i$ iff $v_i(S) \geq \theta_i$. 
\end{definition}

The following observation follows from \Cref{def:stronfEEFXShare,def:stronglyEEFX}.

\begin{observation}\label{obs:eefxShareGivesEEFX}
    If $S \subseteq M$ is a \emph{strongly EEFX-feasible} set for agent $i$, then it is also EEFX-feasible for agent $i$. 
\end{observation}
In our discussion, we focus on strongly EEFX-feasible sets rather than EEFX-feasible sets. Although this is a stronger notion, the fact that strongly EEFX feasibility is defined purely in terms of shares, makes it easier to deal with and connect it with RMMS. This connection in turn gives us a way to prove compatibility of EEFX with EF1.

\paragraph{Non-degeneracy.} 
An instance $\mathcal{I}$ is \emph{non-degenerate} if, for every agent $i\in N$, the valuations of distinct sets are distinct, i.e., $v_i(S) \not= v_i(T)$ whenever $S \not= T$. Chaudhury et al.~\cite{chaudhury2024efx} have shown that the existence of an EFX allocation for non-degenerate instances implies existence in general. We repeat their argument for completeness. 

\begin{lemma}[\cite{chaudhury2024efx}]\label{lem:non-degenrate} If an allocation satisfying both EF1 and EEFX exists for all non-degenerate instances, then such an allocation exists for all instances. \end{lemma}
\begin{proof} Let $\{v_1, \dots, v_n\}$ be a set of additive valuation functions, possibly degenerate, and let $M = \{g_1,\ldots,g_m \}$ be the set of items. Without loss of generality, we may assume  $v_i$'s to be integer-valued. For any $i \in N$ and any $S \subseteq M$, let us define the valuation function $\bv_i$ as
\[     \bv_i(S) = v_i(S) + \sum_{g_j \in S} \varepsilon^j,\]
where $\varepsilon > 0$ is small enough such that $\sum_{g_j \in M} \varepsilon^j < 1$ For example, $\varepsilon = \frac{1}{2}$ will do. Then $\bv_i(S) \not= \bv_i(T)$ for $S \not= T$, and $v_i(S) \not= v_i(T)$ implies $v_i(S) < v_i(T)$ iff $\bv_i(S) < \bv_i(T)$. 

Consider now any allocation $(X_1,\ldots,X_n)$ with nice properties with respect to the valuation functions $\bv$. The allocation has the same nice properties with respect to the original valuation function. For example, if $S$ and $T$ are disjoint and $\bv_i(S \setminus \{g\}) < \bv_i(T)$ then we have, $v_i(S \setminus \{g\}) \leq v_i(T)$. 
So, EF1 is preserved. Similarly, if there is a partition $(P_1,\ldots, P_{n-1})$ of $\cup_{j \in [n]\setminus \{i\}} X_j$ such that agent $i$ does not strongly envy any of the $Y_j$ with respect to $\bv_i$, then they do not strongly envy any of them with respect to the $v_i$ as well, implying EEFX is preserved as well. 
\end{proof}

\subsection{Lone Divider Approach} \label{sec:lone}
When the goal is to achieve a share-based notion, typically an approximation of maximin share ($\alpha$-MMS), a celebrated technique is that of \emph{lone divider} \cite{AIGNERHOREV2022164,EpistemicMonotone,amanatidis2021maximum,Hummel2025,akrami2025matroidsequitable}. For an agent $i \in N$, let us write $\alpha_i$ to denote her fair share value. This algorithm runs in rounds, and in each round a subset of agents gets allocated one bundle each that they find fair according to their own fair share. These agents leave the process, and if there are still some agents remaining, the algorithm proceeds to satisfy them. Let $n'$ be the number of remaining agents at the beginning of some round of this algorithm. Let $i$ be one such remaining agent. Then $i$ divides the set of remaining items $M'$ into $n'$ bundles $B_1, \ldots, B_{n'}$ each of value $\alpha_i$ for $i$; we say that agent $i$ has the role of the \emph{lone divider} in this round. Then, a bipartite graph is constructed with one part corresponding to these $n'$ bundles and the other part corresponding to the remaining agents. There is an edge between an agent $j$ and a bundle $B_k$, iff $v_j(B_k) \geq \alpha_j$. A matching in this graph is called \emph{envy-free} if there is no edge between unmatched agents and matched bundles, i.e.,  all the unmatched agents value all the matched bundles below their fair share. It can be proven that an envy-free matching always exists. Then, under one such matching, the matched agents will be allocated their matched bundles and leave the algorithm, and the algorithm proceeds to the next round. See \Cref{alg:lone-divider}.

\begin{algorithm}[t]
\caption{Lone Divider \cite{amanatidis2021maximum}
\\ \textbf{Input:} $n$ additive valuations $\{v_i\}_{i\in N}$ over items in $M$ and fair shares $\{\alpha_i\}_{i\in N}$. 
\\ \textbf{Output:} Allocation $A = (A_1, \ldots, A_n)$}
\label{alg:lone-divider}
\SetAlgoLined
\DontPrintSemicolon
\LinesNumbered

$N' \leftarrow N$ \;

\While{$N' \neq \emptyset$}{
    Let $i \in N'$ \;
    Let $B=(B_1,B_2,\ldots, B_{n'})$ be such that $v_i(B'_j) \geq \alpha_i$ for all $j \in [n']$\; 

Consider the bipartite graph $G=(V=N'\times B,E)$, where $E=\{(j,B_k)\ : \ v_j(B_k)\ge \alpha_j\}$\;

    Find an envy-free matching between bundles in $B$ and agents in $N'$\;
    \ForEach{$j\in N'$ matched to some $B_k\in B$}{
        $A_j\gets B_j$\;
    }
    Update $N'$ to be the set of unallocated agents and $M \leftarrow M\setminus\left(\bigcup_{i\in N\setminus N'}A_i\right)$ to be the set of unallocated items\;
}
\end{algorithm}

In order to prove the correctness of this algorithm, one needs to show that any remaining agent $i$ can play the role of the lone divider. Equivalently, one needs to show that $i$ can divide the set of remaining items into $n'$ many bundles of value at least $\alpha_i$ (to agent $i$), assuming that $n-n'$ many bundles of value strictly less than $\alpha_i$ have been removed in the previous rounds. Clearly, whether this holds or not depends on $\alpha_i$.

\section{Properties of RMMS}\label{sec:RMMS}
In this section, we study the residual maximin share (RMMS) in greater detail and establish several of its key properties. These properties provide further insight into the behavior of RMMS and will play an important role in our subsequent results.

\begin{observation}
    \label{obs:MXS}
    For every monotone valuation, the RMMS is at least as large as the MXS.
\end{observation}
\begin{proof}
    The proof of Lemma 3.1 in~\cite{EpistemicMonotone} shows that removing a single bundle whose value is smaller than the MXS and reducing the number of agents by one, the value of the MXS share cannot decrease.  This in combination with the fact that the MXS is self-feasible (its value is never larger than the MMS, which is self-feasible) implies that the MXS is residual self-feasible. Among all shares that satisfy the residual self-feasibility property, RMMS has the highest value. 
\end{proof}
Being residual self-feasible implies that a share is feasible. This was demonstrated in several special cases. see for example~\cite{amanatidis2017approximation, AkramiRathi2025simultaneous}.

\begin{observation}
\label{obs:RSFfeasible}
    The RMMS is {\em feasible} for the class of monotone valuations. 
\end{observation}

The proof of~\Cref{obs:RSFfeasible} is omitted, because in~\Cref{obs:EFXRSF} we will prove an even stronger statement.

Hence like MXS, the RMMS is feasible for all monotone valuations. It enjoys the advantage of being at least as large as the MXS. Here, we establish another aspect in which RMMS is preferable over MXS. MXS is a somewhat unnatural share, as it is not monotone -- increasing the values of some bundles might decrease the MXS share value (see \Cref{example:mxs}). 
In contrast, RMMS is monotone, and moreover, is {\em self-maximizing}~\cite{BF25} (see~\Cref{def:selfMaximizing}), a property that implies several other natural properties (such as monotonicity, and being 1-Lipschitz). 

\begin{proposition}
\label{pro:RSFselfMaximizing}
    The RMMS is a {\em self-maximizing} share. 
\end{proposition}
\begin{proof}
Consider an arbitrary partition of the set $2^{\items}$ of all $2^m$ possible bundles into two subsets, $Y$ and $N$. We say that $Y$ is self-feasible if there is an $n$-partition of $\items$ in which all $n$ bundles are in $Y$. We say that $Y$ is residual self-feasible if for ever $0 \le k < n$ and every $k$ bundles in $N$, removing their items from $\items$, the remaining set of items has an $(n-k)$-partition in which all $n-k$ bundles are in $Y$. 

Given a valuation function $v$, the RMMS value partitions the set $2^{\items}$ of all $2^m$ possible bundles into the subset $Y_v$ of those bundles that are acceptable (have $v$ value at least as high as the share value), and the subset $N_v$ of those bundles that are not acceptable. $Y_v$ is required to be residual self-feasible, and among all possible residual self-feasible $Y$, RMMS dictates that we choose the one in which the lowest value bundle has highest possible value (under $v$). Hence no matter what $v'$ an agent with valuation $v$ reports, the corresponding set of acceptable bundles under $v'$ will necessarily contain at least one bundle $T$ of value $v(T) \le \rmms(\items, v, n)$.
\end{proof}

In contrast to RMMS, the following example shows that the notion of the \emph{strong EEFX share} is not self-maximizing.

\begin{example}
    Consider an instance with two agents and four items $\{g_1,g_2,g_3,g_4\}$, where
    $$
    v(g_1)=v(g_2)=4 \quad \text{and} \quad v(g_3)=v(g_4)=2.
    $$
    The strong EEFX share is $\theta=4$, attained by the bundle $\{g_3,g_4\}$.    Now consider a new valuation function $v'$ obtained by decreasing the value of $g_4$ from $2$ to $1$, i.e.,
    $$
    v'(g_1)=v'(g_2)=4,\qquad
    v'(g_3)=2,\qquad
    v'(g_4)=1.
    $$
    Under $v'$, the bundle $\{g_3,g_4\}$ is no longer strong EEFX-feasible, and the strong EEFX share increases to $\theta'=5$. Thus, decreasing the value of an item can increase the strong EEFX share, showing that the strong EEFX share is not self-maximizing. 
\end{example}

We further observe that a proof technique of~\cite{AkramiRathi2025simultaneous} implies feasibility of RMMS simultaneously with satisfying a strong comparison-based fairness notion.

\begin{algorithm}[t]
\caption{Lone Divider + Envy-based guarantees \cite{AkramiRathi2025simultaneous}
\\ \textbf{Input:} $n$ additive valuations $\{v_i\}_{i\in N}$ over items in $M$ and fair shares $\{\alpha_i\}_{i\in N}$. 
\\ \textbf{Output:} Allocation $A = (A_1, \ldots, A_n)$}
\label{alg:lone-divider-efx}
\SetAlgoLined
\DontPrintSemicolon
\LinesNumbered

Set $r=1$, $N_r = N$ \;

\While{$N_r \neq \emptyset$}{
    Let $i \in N_r$ \;
    Let $P=(P_1,P_2,\ldots, P_{n_r})$ be such that $v_i(P_j) \geq \alpha_i$ for all $j \in [n']$\; 
    \ForEach{$k \in [n_r]$}{
        Let $P'_k \subseteq P_k$ be inclusion-wise minimal such that there exists $j \in N_r$ with $v_j(P'_k) \geq \alpha_j$ \;
    }
    \If{there exists $j \in N \setminus N_r$ and $k \in [n_r]$ such that $j$ EFX-envies $P'_k$}{
        Let  $S \subseteq P'_k$ be inclusion-wise minimal such that there exists $j \in N \setminus N_r$ with $v_j(S) > v_i(A_j)$ \;
        $M \leftarrow (M \setminus S) \cup A_j$ \;
        $A_j \leftarrow S$ \;
    }
    \Else{
    Consider the bipartite graph $G=(V=N\times B,E)$, where $E=\{(j,P'_k)\ : \ v_j(P'_k)\ge \alpha_j\}$\;

    Find an envy-free matching between bundles in $P'$ and agents in $N_r$\;
    \ForEach{$j\in N_r$ matched to some $P'_k$}{
        $A_j\gets P'_j$\;
    }
    $r \rightarrow r+1$ \;
    Set $N_r$ to be the set of agents with no bundle and $M \leftarrow M\setminus\left(\bigcup_{i\in N\setminus N_r}A_i\right)$ to be the set of unallocated items\;
    }
    
}
\end{algorithm}
\begin{lemma}
\label{obs:EFXRSF}
    In every allocation instance with monotone valuations, there is a partial RMMS allocation that is EFX.
\end{lemma}
The proof is similar to the proof of Akrami and Rathi~\cite{AkramiRathi2025simultaneous} that for additive valuations, there always are allocations that are $\frac{2}{3}$-MMS and EFX. For completeness, we present the proof.

\begin{proof}
We say that a bundle $S$ is \emph{desirable} for agent $a_i$ if $v_i(S) \ge \rmms(\items, v_i, n)$. 

Observe that the RMMS value of an agent might be~0, if fewer then $n$ items have positive value for the agent. This somewhat complicates our presentation, adding distinctions between empty and non-empty bundles.

We present an allocation algorithm that proceeds in rounds. See \Cref{alg:lone-divider-efx} for the pseudocode.

At the beginning of every round, we have {\em assigned} items (those currently assigned to agents) and {\em free} items (the remaining items). Every agent is either {\em wealthy} (already holds a desired bundle) or {\em poor} (holds no item at all). The algorithm ends when all agents are wealthy. (\Cref{pro:AR24} will show that indeed such a stage is reached.)

Initially, all items are free. We now describe a single round $r$.

\begin{enumerate}
    \item Let $n_r \ge 1$ denote the number of poor agents, and let $i$ be a poor agent. Let $F_r$ denote the set of free items. Construct an arbitrary $v_i$-acceptable $n_r$-partition of $F_r$. This is a collection ${\mathcal P} = (P_1, \ldots, P_{n_r})$ of $n_r$ disjoint bundles such that every $P_j$ is desirable for $i$.  (\Cref{pro:AR24} will show that such a partition must exist.)
    \item For each non-empty bundle $P_j$ in $\mathcal P$, let $P'_j \subseteq P_j$ be a minimal non-empty subset  of $P_j$ (minimal in the sense that no strict subset of $P'_j$ qualifies) so that at least one poor agent desires $P'_j$. Such a $P'_J$ necessarily exists, because $i$ desires $P_j$. If there are several possible candidates for $P'_j$, choose one of them arbitrarily. 
    \item If for some bundle $P'_j$ there is a strict subset $S \subset P'_j$ that some wealthy agent values strictly more than her current bundle, then let $S$ be a minimal such subset. The items currently assigned to the respective wealthy agent become free, and instead $S$ is assigned to that agent. This ends the round. 
    \item Else, construct a bipartite graph with poor agents on one side, bundles $P'_1, \ldots, P'_{n_r}$ on the other side, and edges between agents and bundles that they desire. 
    \begin{enumerate}
        \item If this graph has a perfect matching, allocate the bundles to the poor agents according to any such matching, and the algorithm ends.
        \item Else, let $t_r$ be the smallest integer such that there is a set ${\mathcal T}$ of $t_r$ bundles such that there are only $t_r - 1$ poor agents that desire a bundle from ${\mathcal T}$. By Hall's condition, $t_r \le n$, and by the fact that every bundle is desired by at least one poor agent, $t_r \ge 2$. Match each of these $t_r - 1$ poor agents to some bundle in ${\mathcal T}$ (such a matching, leaving one bundle unmatched, must exist, by Hall's theorem and the minimality of $t_r$). Assign items to agents according to this matching, and end the round.
    \end{enumerate}
\end{enumerate}

We now prove the termination and correctness of the algorithm.

\begin{claim}
\label{pro:AR24}
    The algorithm described above terminates.
\end{claim}

    Observe that at every round $r$, for every agent $a_i$ that is poor at the beginning of the round it holds that no bundle that was allocated in a previous round is desirable. The residual self-feasibility property of RMMS then implies that an RMMS partition of $F_r$ exists, as required in step~1 of the algorithm. 

    In every round, either a wealthy agent replaces her bundle by one of higher value (step~3), or at least one poor agent becomes wealthy (step~4). As there are only finitely many different bundles that can be given to an agent (at most $2^m$), step~3 can be executed at most finitely many times (less than $n \cdot 2^m$). Step~4 can be executed at most $n$ times before all agents are wealthy, and then the algorithm ends.

\begin{claim}
    The partial allocation is RMMS and EFX.
\end{claim}

    When the algorithm ends, all agents are wealthy, implying that the partial allocation is RMMS.

    Every bundle containing at least two items that is allocated by the algorithm is minimal, in the sense that removing any item from it, no poor agent desires it, and no wealthy agent envies it. 
    (The reason for excluding in the above statement bundles with only one item is because there might be agents whose RMMS value is~0. If so, bundles containing a single item might not be considered to be minimal.) As every agent gets a desirable bundle, removing any item from the bundle of any other agent, the remaining bundle is not envied. Hence, the partial allocation is EFX.

This completes the proof of~\Cref{obs:EFXRSF}.
\end{proof}

\propFeige*

The proposition follows immediately from applying the following~\Cref{lem:complete} to the partial allocation of~\Cref{obs:EFXRSF} (and using the fact that EFX implies EFL).

\begin{lemma}
    \label{lem:complete}
    Let ${\mathcal P} = (P_0, P_1, \ldots, P_n)$ be a partial allocation that is EFL. Then there is a full allocation ${\mathcal A} = (A_1, \ldots, A_n)$ that is EFL, and satisfies $v_i(A_i) \ge v_i(P_i)$ for every agent $a_i$. 
\end{lemma}

The known approach for proving statements similar to~\Cref{lem:complete} is to use the LMMS algorithm and apply~\Cref{lem:LMMS}. Trivially, this implies a variation on~\Cref{lem:complete}, in which the final allocation is EF1 (but not necessarily EFL). Moreover, if in the initial EFL partial allocation no agent envies $P_0$ (as is the case for example in the partial EFX allocation of~\cite{chaudhury2021little}), then the final allocation will be EFL.

Our proof of~\Cref{lem:LMMS} is designed in a way that will later allow us to prove also~\Cref{pro:complete}. 
We add a pre-processing step, in which the partial allocation ${\mathcal P} = (P_0, P_1, \ldots, P_n)$ of~\Cref{lem:complete} is replaced by a new partial allocation ${\mathcal P'} = (P'_0, P'_1, \ldots, P'_n)$. The partial allocation ${\mathcal P'}$ will have the property that  there is no free item (item in $P'_0$) that some agent strictly prefers over her own bundle. 

We now present the proof of~\Cref{lem:complete}. 

\begin{proof}
We perform a pre-processing step that proceeds in rounds. In every round, if there is a free item that some agent strictly prefers over her own bundle, one such agent (chosen arbitrarily) takes that item, and gives back her other items (making them free). When every agent weakly prefers her current bundle over every single  free item, the pre-processing step ends. 
The number of rounds is less than $n\cdot m$, even in the most naive implementation of the pre-processing step.

Observe that after the pre-processing step, the resulting partial allocation ${\mathcal P'} = (P'_0, P'_1, \ldots, P'_n)$ is still EFL and satisfies $v_i(P'_i) \ge v_i(P_i)$ for every agent $a_i$.
We use ${\mathcal P'}$ instead of ${\mathcal P}$ as the starting point for the LMMS algorithm.

~\Cref{lem:LMMS} implies that after running the LMMS algorithm, the final allocation ${\mathcal A} = (A_1, \ldots, A_n)$ has the EFL property. Indeed, consider any two agents, $a_i$ and $a_j$. If $A_j$ is a bundle that some agent held in ${\mathcal P'}$, 
then by not having EFL envy in ${\mathcal P'}$, $a_i$ has no EFL envy towards $a_j$ in the ${\mathcal A}$. Else, if $A_j$ contains items that were added by the LMMS algorithm, let $g_j$ be the last item to be added to $A_j$. Then $v_i(A_i) \ge v_i(g_j)$ (because of the pre-processing step), and $v_i(A_i) \ge v_i(A_j \setminus g_j)$ (by ~\Cref{lem:LMMS}). Hence, the EFL property holds.
\end{proof}

For some natural restricted classes of valuations, the RMMS offers nontrivial guarantees with respect to the MMS. A valuation $v$ is \emph{unit-demand} if the value of every bundle $S \subseteq M$ is equal to the value of its most valuable item, i.e., $v(S) = \max_{g \in S} v(g)$.

\begin{observation}
\label{obs:RSFdominating}
    For unit-demand valuations, the RMMS equals the MMS, for additive valuations it is at least $\frac{2}{3}$-MMS (the ratio $\frac{2}{3}$ can be replaced by $\frac{2n}{3n-1}$ for odd $n$ and $\frac{2n-2}{3n-4}$ for even $n$), and for subadditive valuations, it is at least $\frac{1}{n}$-MMS. 
\end{observation}
\begin{proof}
For unit-demand valuations, it is not difficult to see that both the MMS and the RMMS are equal to the value of the $n$th most valuable item.

The fact that for additive valuations, $\frac{2}{3}$-MMS  (more precisely, $\frac{2n}{3n-1}$ for odd $n$ and $\frac{2n-2}{3n-4}$ for even $n$) has the residual self-maximizing property was proved in~\cite{KPW18}.

For subadditive valuations, we prove that every EF1 allocation gives each agent at least $\frac{1}{n}$-MMS. Using \Cref{prop:Feige}, it follows that MXS implies $\frac{1}{n}$-MMS, and consequently that also RMMS implies $\frac{1}{n}$-MMS.

Our proof is basically the same as the proof given for the additive case~\cite{caragiannis2019unreasonable}. 
Consider an arbitrary EF1 allocation $\mathcal{A} = (A_1, \ldots, A_n)$, and suppose that agent $a_n$ has a subadditive valuation $v_n$. For every $1 \le j \le n-1$, if $A_j$ is not empty then there is an item $g_j \in A_j$ such that $v_n(A_n) \ge v_n(A_j \setminus g_j)$. Consider now any MMS partition ${\mathcal P} = (P_1, \ldots, P_n)$. At least one of the parts does not contain any of the $g_j$s. Its value is then at most $v_n(\items \setminus \{g_1, \ldots, g_{n-1}\}) \le A_n + \sum_j v_n(A_j \setminus g_j) \le n\cdot v_n(A_n)$, where the first inequality uses subadditivity of $v_n$.
\end{proof}

In comparison, for additive valuations the value of MXS is at least $\frac{4}{7}$-MMS, but sometimes smaller than $\frac{3}{5}$-MMS~\cite{ABM18}. 

In passing, we note that the combination of~\Cref{obs:RSFfeasible} and~\Cref{obs:RSFdominating} 
implies the following corollary.

\begin{corollary}
    \label{cor:subadditive}
    For agents with subadditive valuations,  $\frac{1}{n}$-MMS allocations always exist.
\end{corollary} 

Our proof for~\Cref{obs:RSFdominating} actually shows that every EF1 allocation gives agents with subadditive valuations at least $\frac{1}{n}$-MMS. Hence,~\Cref{cor:subadditive} is also a consequence of the existence of EF1 allocations~\cite{lipton2004approximately}, or the existence of MXS allocations~\cite{EpistemicMonotone}. We are not aware of the corollary being stated in previous work.
The bound stated in~\Cref{cor:subadditive} holds for all $n$. It is known that for subadditive valuations, $\rho$-MMS is not feasible for $\rho > \frac{1}{2}$~\cite{GhodsiHSSY22}, $\frac{1}{2}$-MMS is feasible for $n \le 4$~\cite{cCMS25}, and $\frac{1}{14 \log n}$-MMS is feasible for all $n$~\cite{FH25}.
(Recent work improves the ratio to $\frac{1}{8 \log \log n}$~\cite{SS25,feige25}.) 

The RMMS satisfies three desirable properties: it is feasible (for the class of all monotone valuations), it is self-maximizing, and it has reasonably high value, in particular, at least $\frac{2}{3}$-MMS for additive valuations. We are not aware of any other share that is known to enjoy such a combination of properties. 

This is not to say that the RMMS has no weaknesses. The next subsections present properties of RMMS that can be regarded as weaknesses.

\subsection{RMMS Beyond Additive Valuations} \label{sec:beyond}

Recall that the maximin share (MMS) is the value of the worst bundle in the best $n$-partition. One may consider a related share, the {\em minimax share}, defined as the value of the best bundle in the worst $n$-partition. For additive valuations, the minimax share is at least as high as the maximin share, but for other classes of valuations this need not hold. In particular, for submodular valuations (satisfying $v(S + e) - v(S) \ge v(T + e) - v(T)$ for every item $e$ and sets $S \subset T \subset \items$), the value of the minimax share might be as low as $\frac{1}{n}$ times the MMS. This happens when there are $n^2$ items that belong to $n$ groups, each containing $n$ items, and the value of a set $S$ is the number of groups that it intersects. The MMS is $n$ (partition $\items$ into $n$ bundles, each containing one item from each group), whereas the value of the minimax share is only~1 (partition $\items$ into $n$ bundles, each composed of a single group).

\begin{proposition}
\label{pro:minimax}
    For every class of valuations, the value of the RMMS is not larger than that of the minimax share.
\end{proposition}

\begin{proof}
    Consider an arbitrary $n$-partition of $\items$ (or specifically, the one that determines the value of the minimax share). The RMMS value cannot be strictly higher than that of the most valuable bundle in the partition. Otherwise, for RMMS, we would be allowed to remove any $n-1$ bundles of that partition (as their value would be strictly smaller than the RMMS), and the remaining bundle would have to have value at least the RMMS.  
\end{proof}

Proposition~\ref{pro:minimax} implies
that beyond additive valuations, the RMMS does not approximate the MMS very well. 
In particular, for submodular valuations RMMS only ensures $\frac{1}{n}$-MMS, whereas $\frac{10}{27}$-MMS allocations are known to exist~\cite{BUF23}.

\subsection{Computational Aspects}\label{sec:comp}

As is standard when considering arbitrary valuation functions, we assume query access to the valuations. A common query model is that of value queries, in which a query specifies a bundle $S$, and the agent $a_i$ replies with its value $v_i(S)$. An even simpler kind of query model is that of comparison queries, a query model introduced in~\cite{bu2024fair}. In a comparison query, two bundles $S$ and $T$ are presented to an agent $a_i$, and the agent only needs to reply with a single bit, indicating whether $v_i(S) \ge v_i(T)$ or not. Observe that two value queries suffice in order to implement a comparison query, but comparison queries cannot implement value queries.

We assume that values of all bundles are integers, in the range $[0, K]$, implying, in particular, that $v_i(\items) \le K$. We refer to an algorithm (working under some query model to the valuations, where replying to a query is assumed to take one unit of time) as strongly polynomial time if it runs in time polynomial in $n$, $m$ (independent of $K$), as polynomial time if it runs in time polynomial in $n$, $m$ and $\log K$, and as pseudo-polynomial time if it runs in time polynomial in $n$, $m$ and $K$. A computational problem that is weakly NP-hard does not have polynomial time algorithms (unless P=NP), but may have pseudo-polynomial time algorithms. 

We now discuss the computational  complexity of tasks related to the current paper.

\begin{proposition}
\label{pro:complete}
    The algorithm proving~\Cref{lem:complete} runs in strongly polynomial time, using only comparison queries.
\end{proposition}

\begin{proof}
    Inspection shows that the pre-processing step in the proof of~\Cref{lem:complete} runs in polynomial time and can implemented using only comparison queries. The same applies to the LMMS algorithm.
\end{proof}

In contrast, computing the RMMS value is not as easy.

\begin{proposition}
    \label{pro:NPhard}
    Computing the RMMS value for additive valuations is weakly NP-hard. Likewise, computing  feasible $n$-partitions is weakly NP-hard.
\end{proposition}

\begin{proof}
    For $n=2$ and additive valuations, the RMMS equals the MMS, which is weakly NP-hard to compute (by an immediate reduction from the weakly NP-hard problem of {\em partition}). Likewise, finding a feasible partition is weakly NP-hard. The case of larger $n$ can be reduced to from the case of two agents, by adding $n - 2$ items of exceptionally high value. 
\end{proof}

\section{Existence of EEFX+EFL Allocations for Additive Valuations} \label{sec:eefx-efl}

As is standard in the EFX literature, to simplify exposition, we restrict our attention to non-degenerate instances $\instance$ without loss of generality (by \Cref{lem:non-degenrate}).

We proposed the concept of residual maximin share (RMMS) (\Cref{def:rmms}) and proved that partial allocations that are EFX and RMMS exist as well as complete allocations that are EFL and RMMS (\Cref{prop:Feige}). Hence, in order to prove the existence of complete allocations that are EFL and EEFX at the same time (and therefore also EF1 and EEFX), it suffices to prove that $\theta_i \leq \rmms(M,v_i,n)$ or that $\theta_i$ is \emph{residual self-feasible}. In other words, for every $0 \leq k < n$, removing $k$ bundles each of value (strictly) less than $\theta_i$ for agent $i$, there is an $(n-k)$ partition of the set of remaining items, where each part has value at least $\theta_i$ for agent $i$. We first handle the case of $k=0$ where no bundles have been removed so far (in \Cref{lem:k-is-0}). Then, we extend the argument to the general case of $0<k<n$ where several low-valued bundles have already been removed (in~\Cref{lem:lone-divider}). 

We now state and prove~\Cref{thm:self-feasible}, which identifies the key property used to derive our main result (\Cref{thm:main}).

\begin{restatable}{theorem}{thmSelfFeasible}\label{thm:self-feasible}
    Given a non-degenerate instance $\instance$, for all $i \in N$, $\theta_i$ is residual self-feasible, i.e., $\theta_i \leq \rmms(M,v_i,n)$.
\end{restatable}

If $\theta_i=0$, then proving $\theta_i$ is residual self-feasible is trivial. Thus, for the rest of this section, fix an agent $i$, assume $\theta_i>0$, and let $T$ be the largest-valued EEFX-infeasible bundle for agent $i$. 
Note that since $\theta_i>0$, $\emptyset$ is EEFX-infeasible and thus $T$ exists.
Moreover, since $\I$ is non-degenerate, $T$ is unique. 
By the definition of $\theta_i$ and $T$, we have $v_i(T) < \theta_i$, and there is no bundle $B$ such that $v_i(T) < v_i(B) < \theta_i$.
We therefore obtain the following observation, which will be useful later.

\begin{observation}\label{obs:obv}
    For all $B \subseteq M$, $v_i(B) > v_i(T) \iff v_i(B) \geq \theta_i$. 
\end{observation}

\cite{Caragiannis2023} showed that MMS implies EEFX. Therefore,
$\theta_i \leq \mms_i$ for every agent $i$ (as $\theta_i$ is the smallest share implying EEFX). Since MMS is a self-feasible share, it follows that the strong EEFX share is also self-feasible. For completeness, we provide an alternative proof of the self-feasibility of the strong EEFX share in \Cref{lem:k-is-0}.
\begin{lemma}\label{lem:k-is-0}
    There exists a partition of $(X_1, \ldots, X_n)$ of $M$ such that $v_i(X_\ell) \geq \theta_i$ for all $\ell \in [n]$.
\end{lemma}
\begin{proof}
    Let $(Y_1, \ldots, Y_{n-1})$ be an EFX partition of $M \setminus T$ into $n-1$ bundles for agent $i$. 
    For all $\ell \in [n-1]$, we have
        $$v_i (Y_\ell) \geq \max_{j \in [n-1], g \in Y_j}v_i(Y_j \setminus \{g\}) > v_i(T).$$
    The first inequality follows from $(Y_1, \ldots, Y_{n-1})$ being an EFX partition for agent $i$. The second inequality holds since otherwise $(Y_1, \ldots, Y_{n-1})$ would be an EEFX certificate for $T$, which contradicts the fact that $T$ is EEFX-infeasible. Now, let $j^* = \arg\max_{j \in [n-1], g \in Y_j}v_i(Y_j \setminus \{g\}) \in [n-1]$. Choose $g^* \in Y_{j^*}$ such that $v_i(Y_{j^*} \setminus \{g^*\}) > v_i(T)$. Let $X_1 = T \cup \{g^*\}$, $X_{j^*+1} = Y_{j^*} \setminus \{g^*\}$ and $X_{\ell+1} = Y_\ell$ for all $\ell \in [n-1] \setminus \{j^*\}$. 
    Since the instance is non-degenerate, we have $v_i(X_1) = v_i(T\cup \{g^*\}) > v_i(T)$.
    Thus, for all $\ell \in [n]$, we have $v_i(X_\ell) >  v_i(T)$. By \Cref{obs:obv}, we have for all $\ell \in [n]$, $v_i(X_\ell) \geq \theta_i$. Hence, the lemma follows.
\end{proof}

\begin{lemma}\label{lem:lone-divider}
    Fix an agent $i$ and $k \in [n]$. Let $S_1, \ldots, S_k \subseteq M$ be $k$ bundles such that
    \begin{itemize}
        \item $S_j \cap S_{j'} = \emptyset$ for all $j \neq j'$, and
        \item $v_i(S_j) < \theta_i$ for all $j\in [k]$.
    \end{itemize}
    Then, there exists a partition $(X_1, \ldots, X_{n-k})$ of $M \setminus (\bigcup_{j \in [k]} S_j)$ such that $v_i(X_j) \geq \theta_i$ for all $j \in [n-k]$.
\end{lemma}

\begin{figure}
    \centering
    \includegraphics[width=0.63\linewidth]{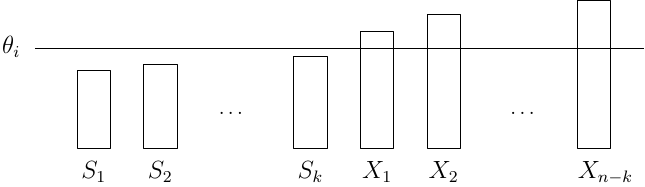}
    \caption{Visualization of the bundles in \Cref{lem:lone-divider}.}
    \label{fig:1}
\end{figure}

See \Cref{fig:1} for a visualization of the bundles in \Cref{lem:lone-divider}. 

To prove \Cref{lem:lone-divider}, we first establish two intermediate results,~\Cref{lem1} and \ref{lem2}. Recall that $T$ is the largest-valued EEFX-infeasible bundle for agent $i$.  Let $S_1, \ldots, S_k \subseteq M$ be defined as in \Cref{lem:lone-divider}, we write 
\[ S = \bigcup_{j \in [k]} S_j, \quad R = T \cap S, \quad U = T \setminus S.\] 
That is, we write $S$ as the union of $k$ bundles $S_1, \dots, S_k$, which we view as removed bundles, $R$ as the part of $T$ already removed, and $U$ as the remaining part of $T$. See \Cref{fig:2} for a visual depiction of the sets $S, R,$ and $U$.

\begin{figure}
    \centering
    \includegraphics[width=0.33\linewidth]{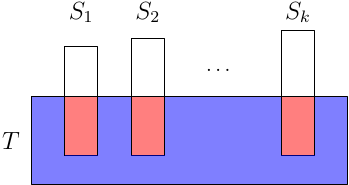}
    \caption{\textcolor{red}{$R = T \cap S$} is the union of the red parts and \textcolor{blue}{$U = T \setminus S$} is the blue part.}
    \label{fig:2}
\end{figure}

\begin{lemma}\label{lem1}
    There exists a partition $(Y_1, \ldots, Y_k)$ of $S$  such that
    \begin{itemize}
        \item $R \subseteq Y_1$, 
        \item $v_i(Y_1) < \theta_i$, and
        \item $v_i(Y_j \setminus \{g\}) < \theta_i$ for all $j \in [k] \setminus \{1\}$ and $g \in Y_j$.
    \end{itemize}
\end{lemma}

See \Cref{fig:3} for a visualization of the bundles in \Cref{lem1}. 

\begin{figure}
    \centering
    \includegraphics[width=0.73\linewidth]{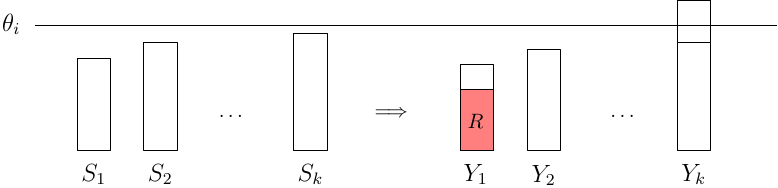}
    \caption{Depiction of $S = \bigcup_{i \in [k]} S_i = \bigcup_{i \in [k]} Y_i$, and $R = S \cap T$ in~\Cref{lem1}.}
    \label{fig:3}
\end{figure}

\begin{proof}
    Let $(Y_1, \ldots, Y_k)$ be such that 
    \begin{enumerate}
        \item $R \subseteq Y_1$, \label{item1}
        \item $v_i(Y_1) < \theta_i$, and \label{item2}
        \item $(Y_2, \ldots, Y_k)$ is an EFX partition of $S \setminus Y_1$ for agent $i$, \label{item3}
    \end{enumerate}
    and subject to \eqref{item1}, \eqref{item2}, and \eqref{item3}, $\max_{j \in [k] \setminus \{1\}, g \in Y_j} v_i(Y_j \setminus \{g\})$ is minimum. Note that such a partition exists because $R$ together with any EFX-partition of $S \setminus R$ for agent $i$ satisfies (1) to (3). This holds true since $R \subseteq T$ and $v_i(T) < \theta_i$. Now, our goal is to show that $(Y_1, \ldots, Y_k)$ satisfies all the conditions stated in \Cref{lem1}.
    Clearly, it satisfies the first two of these conditions. Let $(j^*,g^*) = \arg \max_{j \in [k] \setminus \{1\}, g \in Y_j} v_i(Y_j \setminus \{g\})$.
    It remains to prove $v_i(Y_{j^*} \setminus \{g^*\}) < \theta_i$. 
    If $v_i(Y_{j^*} \setminus \{g^*\}) = 0$, then the claim follows immediately, since  we assume $\theta_i > 0$.
    Hence, we may assume that $v_i(Y_{j^*} \setminus \{g^*\}) > 0$.
    Towards contradiction, assume otherwise:
    \begin{align}
        v_i(Y_{j^*} \setminus \{g^*\}) \geq \theta_i. \label{ineq1}
    \end{align}
    Consider the partition $Y'$ obtained after moving $g^*$ from $Y_{j^*}$ to $Y_1$. That is, $Y'=(Y_1 \cup \{g^*\}, Y_2, \ldots, Y_{j^*} \setminus \{g^*\}, \ldots, Y_k)$. If $v_i(Y_1 \cup \{g^*\}) < \theta_i$, then $Y'$ has properties \eqref{item1}, \eqref{item2}, \eqref{item3}, and we have $\max_{j \in [k] \setminus \{1\}, g \in Y'_j} v_i(Y'_j \setminus \{g\}) < \max_{j \in [k] \setminus \{1\}, g \in Y_j} v_i(Y_j \setminus \{g\})$. The strict inequality follows from non-degeneracy of the instance. This is a contradiction to the choice of $Y$. Hence, 
    \begin{align}
        v_i(Y_1 \cup \{g^*\}) \geq \theta_i. \label{ineq2}     
    \end{align}
    
    \noindent
    Combining inequality~\eqref{ineq1} with the fact that $(Y_2, \ldots, Y_k)$ is an EFX partition of $S \setminus Y_1$ for agent $i$, it follows that for all $j \in [k] \setminus \{1\}$, 
    \begin{align}
        v_i(Y_j) \geq v_i(Y_{j^*} \setminus \{g^*\}) \geq \theta_i. \label{ineq3}
    \end{align}
     We can thus write, 
    \begin{align*}
        k \cdot \theta_i &> \sum_{j \in [k]} v_i(S_j) \tag{$v_i(S_j) < \theta_i$ for all $j \in [k]$}\\
        &= \sum_{j \in [k]} v_i(Y'_j) \tag{$\bigcup_{j \in [k]} S_j = \bigcup_{j \in [k]} Y'_j$}\\
        &= v_i(Y_1 \cup \{g^*\})+ v_i(Y_{j^*} \setminus \{g^*\}) + \sum_{j \in [k] \setminus \{1,j^*\}} v_i(Y_j) \\
        &\geq k \cdot \theta_i, \tag{Inequalities \eqref{ineq1}, \eqref{ineq2}, and \eqref{ineq3}}
    \end{align*}
    which is a contradiction. Therefore, $v_i(Y_{j^*} \setminus \{g^*\}) < \theta_i$, and hence $v_i(Y_j \setminus \{g\}) < \theta_i$ for all $j \in [k] \setminus \{1\}$ and $g \in Y_j$. This completes the proof.
\end{proof}

\begin{lemma}\label{lem2}
    Let $(X_1, \ldots, X_{n-k-1})$ be an EFX partition of $(M \setminus S) \setminus U$ for agent $i$. Then $v_i(X_j) \geq \theta_i$ for all $j \in [n-k-1]$.
\end{lemma}
\begin{proof}
    Towards contradiction, let us assume $v_i(X_{j^*}) < \theta_i$ for some $j^* \in [n-k-1]$. Then, by \Cref{obs:obv}, we have $v_i(X_{j^*}) \leq v_i(T)$. Let $(Y_1, \ldots, Y_k)$ be a partition of $S$, as described in \Cref{lem1}. We will show that $Z = (Y_1 \setminus R, Y_2, \ldots, Y_k, X_1, \ldots, X_{n-k-1})$ is an EEFX certificate for $T$, and hence leading to a contradiction.

    To begin with, note that $Z$ is a partition of $M \setminus T$ into $n-1$ bundles. By \Cref{lem1}, we have $v_i(Y_1 \setminus R)\le v_i(Y_1) < \theta_i$, and hence, by~\Cref{obs:obv}, we have 
    \begin{align}
         v_i(Y_1 \setminus R) \le v_i(T). \label{ineq4}
    \end{align}
    Moreover, for all $j \in [k]\setminus \{1\}$ and $g \in Y_j$, we have $v_i(Y_j \setminus \{g\}) < \theta_i$. Using~\Cref{obs:obv} again, we obtain 
    \begin{align}
        v_i(Y_j \setminus \{g\}) \leq v_i(T) \ \text{for all} \ j \in [k]\setminus \{1\} \ \text{and} \ g \in Y_j. \label{ineq5}
    \end{align}
    Since $(X_1, \ldots, X_{n-k-1})$ is an EFX partition of $(M \setminus S) \setminus U$ for agent $i$,  we have $v_i(X_j \setminus \{g\}) \leq v_i(X_{j^*})$ for all $j \in [n-k-1]$ and $g \in X_j$. Now, using our assumption that $v_i(X_{j^*})< \theta_i$ and~\Cref{obs:obv}, we obtain 
    \begin{align}
         v_i(X_j \setminus \{g\}) \leq v_i(T)  \ \text{for all} \ j \in [n-k-1] \ \text{and} \ g \in X_j. \label{ineq6}
    \end{align}
 Combining \eqref{ineq4}, \eqref{ineq5}, and \eqref{ineq6}, we have that $Z = (Y_1 \setminus R, Y_2, \ldots, Y_k, X_1, \ldots, X_{n-k-1})$ is an EEFX certificate for $T$, contradicting the EEFX infeasibility of $T$. This completes our proof.
\end{proof}

\begin{figure}
    \centering
    \includegraphics[width=0.65\linewidth]{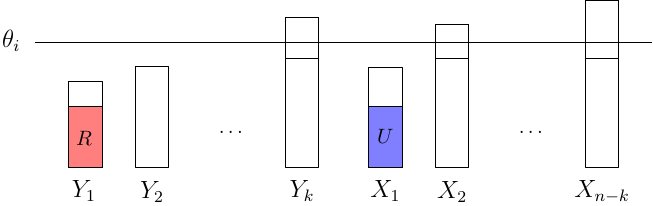}
    \caption{Visualization of the partitions $(Y_1, \ldots , Y_k)$ of $S$ and $(X_1, \ldots , X_{n-k})$ of $M\setminus S$. Here, $R\subseteq Y_1$, $U\subseteq X_1$, and $T=R\cup U$.}
    \label{fig:4}
\end{figure}

We are now ready to prove~\Cref{lem:lone-divider}.

\begin{proof}[Proof of \Cref{lem:lone-divider}.] 
    Let $X=(X_1, \ldots, X_{n-k})$ be a partition of $M \setminus S$ such that
    \begin{enumerate}
        \item $U \subseteq X_1$, 
        \label{prop1}
        \item $v_i(X_j) \geq \theta_i$ for all $j \in [n-k] \setminus \{1\}$, and
        \label{prop2}
    \end{enumerate}
    and subject to these, $v_i(X_1)$ is maximum. Note that, by \Cref{lem2}, such a partition exists. We prove $v_i(X_1) \geq \theta_i$. Towards contradiction, assume 
    $v_i(X_1) < \theta_i$.
    By \Cref{obs:obv}, we have $v_i(X_1) \leq v_i(T)$. Let $(j^*,g^*) = \arg \max_{j \in [n-k]\setminus \{1\}, g \in X_j} v_i(X_j \setminus \{g\})$. If $v_i(X_{j^*} \setminus \{g^*\}) \geq \theta_i$, then moving $g^*$ from $X_{j^*}$ to $X_1$ results in a partition that has properties \eqref{prop1} and \eqref{prop2}, while $v_i(X_1 \cup \{g^*\}) > v_i(X_1)$, which is a contradiction to the choice of $X$. Hence, $v_i(X_{j^*} \setminus \{g^*\}) < \theta_i$, and thus $v_i(X_{j^*} \setminus \{g^*\}) \le v_i(T)$.
    This implies that 
    \begin{align}\label{ineq7}
    v_i(X_j \setminus \{g\}) \le v_i(T)
    \quad \text{for all } j \in [n-k]\setminus\{1\} \text{ and } g \in X_j. 
    \end{align}
    
    Let $(Y_1,\ldots,Y_k)$ be the partition of $S$ as in \Cref{lem1}.
    By the third condition of \Cref{lem1}, we have
    \begin{align}\label{ineq8}
    v_i(Y_j \setminus \{g\}) \le v_i(T)
    \quad \text{for all } j \in [k]\setminus\{1\} \text{ and } g \in Y_j. 
    \end{align}
    Recall that $T = R \cup U$. 
    See \Cref{fig:4} for a visualization of $(Y_1, \ldots, Y_k, X_1, \ldots, X_{n-k})$. Consider the following partition $Z$ of $M \setminus T$ into $n - 1$ bundles: \[ Z = (Y_2,\ldots,Y_k,(Y_1 \setminus R) \cup (X_1 \setminus U), X_2,\ldots,X_{n-k}).\]
    Then, 
    \begin{align}\label{ineq9}
    v_i ((Y_1 \setminus R) \cup (X_1 \setminus U)) = v_i (Y_1) - v_i(R) + v(X_1) - v_i(U) = v_i(Y_1) + v_i(X_1) - v_i(T) \leq v_i(T). 
    \end{align}
    Here, the inequality follows from $v_i(Y_1) < \theta_i$ (thus $v_i(Y_1)\le v_i(T)$) by \Cref{lem1}, and from $v_i(X_1) < \theta_i$ (thus $v_i(X_1) \le v_i(T)$) by our assumption.

    Combining \eqref{ineq7}, \eqref{ineq8}, and \eqref{ineq9}, we conclude that $Z$ is an EEFX certificate for $T$, which contradicts the assumption that $T$ is EEFX-infeasible. Therefore, $v_i(X_1)\ge \theta_i$. This shows that there exists a partition $(X_1, \ldots, X_{n-k})$ of $M \setminus S$ such that $v_i(X_j) \geq \theta_i$ for all $j \in [n-k]$; thereby completing our proof.
    \end{proof}
  
Using~\Cref{lem:k-is-0} and \ref{lem:lone-divider}, we can now  prove the crucial property that connects strong EEFX share and RMMS share.

\thmSelfFeasible*
\begin{proof}
    We need to prove for every $0 \leq k < n$, removing $k$ bundles each of value (under $v_i$) strictly less than $\theta_i$, there is an $(n-k)$ partition of the set of remaining items, where each part has value (under $v$) at least $t$. If $\theta_i = 0$, the claim trivially holds since any partition would do. Otherwise, \Cref{lem:k-is-0} proves the case $k=0$ and \Cref{lem:lone-divider} the case $0<k<n$.
\end{proof}

Finally, using~\Cref{thm:self-feasible}, we establish the main result of this work. 
\thmMain*

\begin{proof}
    By \Cref{lem:non-degenrate}, without loss of generality, we assume the instance is non-degenerate. By \Cref{thm:self-feasible}, for all $i \in N$, $\theta_i \leq RMMS(M, v_i, n)$. Hence, by \Cref{prop:Feige}, there exists a complete allocation that is EFL (and hence EF1) and provides every agent with their strong EEFX share. By \Cref{obs:eefxShareGivesEEFX}, this allocation is EEFX as well. This completes our proof.
\end{proof}

To illustrate the limitations of this approach, we provide an example showing that, for submodular valuations, the strong EEFX share can be strictly larger than the RMMS. Consequently, proving the existence of EEFX and EF1 allocations in this setting requires new techniques.

\begin{example}
    Let $n=2$ and $M=\{a,b,c,d,e,f\}$. Suppose all agents have the same valuation function $v$. For every item $g\in M$, let $v(g)=1$. In addition, let 
    $$
    v(\{a,b,c\}) = v(\{d,e,f\}) = 1.
    $$
    The RMMS value is $1$.
    
    Now consider pairs consisting of one item from $\{a,b,c\}$ and one item from $\{d,e,f\}$. Assign value $2$ to every such pair, except for the pair $\{a,d\}$, which has value $3/2$. For every set $S\subseteq M$ with $|S|\ge 3$, define
    $$
    v(S)=\max_{\substack{T\subseteq S\\|T|=2}} v(T).
    $$
    It is straightforward to verify that $v$ is submodular.
    
    Furthermore, the bundle $\{a,d\}$ is not strong EEFX-feasible. Indeed, if one agent receives $\{a,d\}$, then the other agent receives $M\setminus\{a,d\}$, whose value is $2$. After removing any item from the latter bundle, its value remains $2$, whereas $v(\{a,d\})=3/2$. Therefore, every strong EEFX-feasible bundle must have value strictly greater than $3/2$, implying that
    $$
    \theta_1 > \frac{3}{2} > \rmms_1.
    $$
\end{example}

\section{Polynomial Time Algorithm for EF1 and EEFX Allocations for Restricted Additive Valuations} \label{sec:res-add}
In the previous section, we obtained an EEFX+EFL allocation using the algorithm described in \Cref{obs:EFXRSF}, which computes an allocation satisfying RMMS and EFL. As discussed earlier, computing RMMS is weakly NP-hard. We leave the complexity of computing the strong EEFX share as an open problem. Moreover, even if the strong EEFX share could be computed in polynomial time, it is not clear whether the resulting algorithm would run in polynomial time. The main obstacle is the step in which wealthy agents exchange bundles: while each such exchange strictly increases the utility of some agent, the resulting progress guarantee suffices to prove termination, but not polynomial-time termination.

Motivated by this challenge, we leave the question of whether EEFX+EFL allocations can be computed in polynomial time as an open problem. Instead, in this section, we develop a polynomial-time algorithm (\Cref{alg:eefx+ef1:res-add}) that computes an allocation satisfying both $\ef1$ and $\eefx$ for restricted additive valuations (Theorem~\ref{thm:res}). Our approach completely bypasses the lone-divider framework and instead exploits structural properties specific to restricted additive instances.

\begin{definition}
    A fair division instance $\I = (N,M,V)$ has restricted additive valuations, if for all $i \in N$, $v_i$ is additive and there exists a function $u: M \rightarrow \mathbb{R}_{\geq 0}$ such that $v_i(g) \in \{0,u(g)\}$ for all $g \in M$. 
\end{definition}

Given an instance with restricted additive valuations, for all $i \in N$, let $Z_i = \{g \in M \mid v_i(g)=0\}$ be the set of goods that agent $i$ values at $0$. Let us also sort the goods $M = \{g_1, g_2, \dots, g_m\}$ according to the $u$-value, and assume without loss of generality, $u(g_1) \geq u(g_2) \geq \dots u(g_m)$. 

If for some good $g \in M$ and all agents $i \in N$, $v_i(g)=0$, without loss of generality, we assume $u(g)=0$. 

\begin{definition} \label{def:efxp}
    We say that an allocation $X$ is $\efx^+$ iff for all agents $i,j \in N$ and all $g \in X_j \setminus Z_i$, we have $v_i(X_i) \geq v_i(X_j \setminus g)$. Agent $i$ $\efx^+$-envies bundle $X_j$, iff there exists $g \in X_j \setminus Z_i$ such that $v_i(X_i)< v_i(X_j \setminus g)$.
\end{definition}

By definition, $\efx^+$ implies EFL and EF1.

The following lemma states a sufficient condition for $\eefx$-feasibility that will be useful to prove our main result. Note that this lemma holds even for (unrestricted) additive valuations.

\begin{lemma}\label{lem:eefx-feas}
    Given an allocation $A$, and agent $i$ with restricted additive valuations, if $i$ does not $\efx^+$-envy any other agent and there exists $i' \in N \setminus i$ such that $v_i(A_i) \geq v_i(A_{i'})$, then $A_i$ is $\eefx$-feasible for $i$. 
\end{lemma}
\begin{proof}
    Consider the partition of $M \setminus A_i$ into the following $n-1$ bundles. For all $j \in N \setminus \{i,i'\}$, let $A'_j = A_j \setminus Z_i$ and $A'_{i'} = A_{i'} \cup Z_i$. We prove $(A'_1, \ldots, A'_{i-1}, A'_{i+1}, \ldots, A'_n)$ is an $\eefx$-certificate for $A_i$. 
    
    Agent $i$ does not $\efx^+$-envy any other agent under allocation $A$. Thus, for all $j \in N \setminus \{i,i'\}$ and $g \in A'_j = A_j \setminus Z_i$, $v_i(A_i) \geq v_i(A_j \setminus g) = v_i(A'_j \setminus g)$. Also 
    $$v_i(A'_{i'}) = v_i(A_{i'}) \leq v_i(A_i).$$
    Thus, agent $i$ does not $\efx$-envy bundle $A'_j$ for any $j \in N \setminus i$. Therefore, $(A'_1, \ldots, A'_{i-1}, A'_{i+1}, \ldots, A'_n)$ is an $\eefx$-certificate for $A_i$.
\end{proof}

We begin by describing \Cref{alg:efxp:res-add} that finds an $\efx^+$ allocation $X$ for restricted additive valuations in polynomial time. This allocation $X$ then forms the starting point of our main algorithm (\Cref{alg:eefx+ef1:res-add}). \Cref{alg:efxp:res-add} sorts the goods according to their decreasing $u$-values and then allocates goods one by one to a poorest agent who values this good at its $u$-value.

\begin{algorithm}
    \caption{$\mathrm{ALG}: \efx^+$ for Restricted Additive Valuations } \label{alg:efxp:res-add}
    \DontPrintSemicolon
    \textbf{Input:} A fair division instance $\I = (N, M, V)$ with restricted additive valuations\;
    \textbf{Output:} An $\efx^+$ (and $\ef1$) allocation\;
    \BlankLine
    Sort the goods in $M$ such that $u(g_1) \geq u(g_2) \geq \dots u(g_m)$\;
    \For{$g \leftarrow g_1$ to $g_m$}{
    Let $j^* \leftarrow  \arg \min_{i \in N: v_i(g) = u_i(g)} v_i(X_i) $\;
    $X_{j^*} \leftarrow X_{j^*} \cup g$   \Comment{Allocate $g$ to a poorest agent $j^*$ with $v_{j^*}(g) = u(g)$}
    }
    Order the allocation $X$ such that $u(X_1) \leq \dots \leq u(X_n)$\;
    \Return $X$\;
\end{algorithm}

\begin{lemma}\label{lem:efxp}
    \Cref{alg:efxp:res-add} runs in $O(m\log m + nm)$ time and outputs an allocation that is $\efx^+$ (and EF1) for restricted additive valuations. 
\end{lemma}
\begin{proof}
    First we need $O(m \log m)$ to sort the goods. Then we allocate goods one by one ($O(m)$ many iterations) to the poorest agent who values it at its $u$-value which can be determined in $O(n)$ time.
    
    We prove the $\efxp$ claim via induction. The base case is trivial. Let us assume that the partial allocation $X^r$ at some round $r$ is $\efx^+$. Consider the next round that allocates good $g_{r+1}$ to a poorest agent $i \in N$ with $v_i(g) = u(g)$. Let this agent be $j^*$. Consider any agent $i \neq j^*$. Note that the only possible $\efx^+$-envy in $X^{r+1}$ could be from such agents $i$ to $j^*$ with $v_i(g) =u(g)$. But, by the choice of $j^*$, we have
    \begin{align*}
        v_i(X^{r+1}_i) &= v_i(X^{r}_i)\\
        &\geq v_{j^*}(X^{r}_{j^*}) \tag{by the choice of $j^*$}\\
        & \geq v_i(X^{r}_{j^*}) \tag{$X^{r}_{j^*} \cap Z_{j^*} = \emptyset$}\\
        &= v_i(X^{r+1}_{j^*} \setminus g_{r+1}).
    \end{align*}

    Since, $u(g_{r+1}) \leq u(g_k)$ for all $k \leq r$, $i$ does not $\efx^+$-envy $j^*$. Hence, $X^{r+1}$ is $\efx^+$ and by induction, the final output allocation $X$ is $\efx^+$ as well. Since $\efx^+$ implies EF1, the allocation is also EF1.
\end{proof}

From now on, we will denote the output allocation of \Cref{alg:efxp:res-add} as $X$; wherein we order the bundles such that $u(X_1) \leq \dots \leq u(X_n)$. We now show several useful properties of the output allocation $X$ of \Cref{alg:efxp:res-add}. We begin with the following observation.

\begin{observation}\label{obs:u=v}
    For all agents $i \in N$, $v_i(X_i)=u(X_i)$ and $v_i(X_i) \geq v_i(X_1)$.
\end{observation}

\begin{proof}
    \Cref{alg:efxp:res-add} adds an item $g$ to $X_i$ only when $v_i(g)=u(g)$. Hence, for all $i \in N$, we have $v_i(X_i)=u(X_i)$. Since $u(X_1) \leq \dots \leq u(X_n)$, we can now write $v_i(X_i) = u(X_i) \geq u(X_1) \geq v_i(X_1)$.  
\end{proof}

\begin{lemma} \label{lem:prop}
     For all agents $i>1$, we either have \\
     (a) $v_1(X_i) \leq v_1(X_1)$, or \\
     (b) for all $g \in X_i \setminus Z_1, z \in X_i \cap Z_1$, we have $u(z) \leq u(g)$.
\end{lemma}

\begin{proof}
    Fix an agent $i>1$. If statement (a) holds true, the lemma follows. Therefore, let us suppose (a) does not hold true, i.e., $v_1(X_i) > v_1(X_1)$; we will prove that then (b) holds true. Let $g$ be the last item in $X_i \setminus Z_1$ that was added to $X_i$ in \Cref{alg:efxp:res-add}, i.e, $g \in \arg \min_{h \in X_i \setminus Z_1} u(h)$ and all items added after $g$ to $X_i$ are in $Z_1$. Let $Z'_1 \subset Z_1$ be the set of all items added to $X_i$ after $g$. Towards a contradiction, let us assume there exists $z \in X_i \cap Z_1$ such that $u(z) > u(g)$. Hence $z \notin Z'_1$. Let $X'_1$ and $X'_i$ be the bundles of agents $1$ and $i$ respectively, just before \Cref{alg:efxp:res-add} added $g$ to $X'_i$; 
    \begin{align}
        X'_i = X_i \setminus (Z'_1 \cup \{g\}). \label{eq:xi}    
    \end{align}
    Since \Cref{alg:efxp:res-add} allocates goods in decreasing order of $u$-values and $u(z) > u(g)$, $z$ is already allocated to $i$; $z \in X'_i$. \Cref{alg:efxp:res-add} allocates $g$ to an agent with minimum utility who values $g$ at $u(g)$. Since $v_1(g) = u(g)$, we must have 
    \begin{align}
        v_1(X'_1) \geq v_i(X'_i). \label{eq:xprime}    
    \end{align}
    We have,
    \begin{align*}
        v_1(X_1) &\geq v_1(X'_1) \tag{$X'_1 \subseteq X_1$} \\
        &\geq v_i(X'_i) \tag{Equation \eqref{eq:xprime}}\\
        &= v_i(X_i \setminus (Z'_1 \cup g)) \tag{Equation \eqref{eq:xi}} \\
        &= u(X_i \setminus (Z'_1 \cup g)) \tag{\Cref{obs:u=v}} \\
        & \geq u(X_i \setminus (Z'_1 \cup z)) \tag{$u(z) \geq u(g)$ and additivity of $u$}\\
        &\geq v_1(X_i) \tag{$v_1(Z'_1 \cup z)=0$}\\
        &> v_1(X_1), \tag{by our assumption}    
    \end{align*}
    which is a contradiction. Thus, if (a) does not hold true, (b) does.
\end{proof}

\begin{observation}\label{obs:eefx-feas}
    For all agents $i>1$, $X_i$ is $\eefx$-feasible for $i$.
\end{observation}
\begin{proof}
    Fix an agent $i>1$. Using~\Cref{obs:u=v}, we have $v_i(X_i)\geq v_i(X_1)$. Now,~\Cref{lem:efxp} implies that $i$ does not $\efx^+$-envy any other agent. Hence, by~\Cref{lem:eefx-feas}, $X_i$ is $\eefx$-feasible for $i$.
\end{proof}

We now prove the following lemma that describes necessary conditions for $X_1 \cup$ some goods from $Z_1$ to not be $\eefx$-feasible for agent $1$.
\begin{lemma}\label{lem:X1}
    Let $Z'_1 \subseteq Z_1$, $Y_1 = X_1 \cup Z'_1$ and $Y_i = X_i \setminus Z'_1$ for all $i > 1$. If $Y_1$ is not $\eefx$-feasible for agent $1$, then for all $i>1$ the following holds:
    \begin{enumerate}
        \item $v_1(Y_i)>v_1(Y_1)$.
        \item $Y_i \setminus Z_1$ is $\eefx$-feasible for agent $1$.
        \item For all $g \in Y_i \setminus Z_1$ and $z \in Y_i \cap Z_1$, $u(z) \leq u(g)$.
        \item $Y_1 \cup Z_1$ is $\eefx$-feasible for agent $1$. \label{point-4}
    \end{enumerate}
\end{lemma}
\begin{proof}
Note that, by~\Cref{lem:efxp}, agent $1$ does not $\efxp$ envy any other agent in $X$. Hence $1$ does not $\efxp$ envy anyone in $Y$ either.
We prove the points one by one.
\begin{enumerate}
    \item Otherwise, by \Cref{lem:eefx-feas}, $Y_1$ is $\eefx$-feasible for agent $1$.
    \item Agent $1$ does not $\efx^+$-envy any other agent in $Y$. Let $Y'_j = Y_j$ for all $j \in N \setminus\{1,i\}$, $Y'_1 = Y_i \setminus Z_1$, and $Y'_i = Y_1 \cup Z_1$. Since $v_1(Y_i \setminus Z_1)=v_1(Y_i)>v_1(Y_1)=v_1(Y_1 \cup Z_1)$, agent $1$ does not $\efx^+$-envy any other agent in $Y'$. By \Cref{lem:eefx-feas}, $Y'_1 = Y_i \setminus Z_1$ is $\eefx$-feasible for agent $1$.
    \item By the first point, $v_1(Y_i)>v_1(Y_1)$ for all $i \in N$. Note that $v_1(Y_i) = v_1(X_i)$, since the all items outside $Z_1$ are allocated similarly in both allocations. Hence, $v_1(X_i)>v_1(X_1)$. By~\Cref{lem:prop}, for all $g \in X_i \setminus Z_1, z \in X_i \cap Z_1$, we have $u(z) \leq u(g)$. We have $Y_i \setminus Z_1= X_i \setminus Z_1$ and $Y_i \cap Z_1 \subseteq X_i \cap Z_1$. Therefore, also for all $g \in Y_i \setminus Z_1$ and $z \in Y_i \cap Z_1$, $u(z) \leq u(g)$.
    \item By~\Cref{lem:efxp}, for all $i>1$ and $g \in Y_i \setminus Z_1$, $v_1(Y_1) \geq v_1(Y_i \setminus g)$. Therefore $(Y_2 \setminus Z_1, \ldots, Y_n \setminus Z_1)$ is an EEFX certificate for $Y_1 \cup Z_1$. \qedhere
\end{enumerate}
\end{proof}

Let us now describe our main algorithm (\Cref{alg:eefx+ef1:res-add}) that computes an allocation that is $\eefx$ and EF1 in polynomial time. For the corner case of $m \leq n$, it allocates exactly one item to the first $m$ agents. This allocation is EFX and hence both $\efxp$ and EF1. Also in this case, the running time is $O(n)$. From now on, we assume $m>n$.

\Cref{alg:eefx+ef1:res-add} begins with the $\efx^+$ allocation $X$ (computed by~\Cref{alg:efxp:res-add}) and renames it as $Y$. Then, as long as $Y_1$ is not $\eefx$-feasible for agent $1$, we find a \emph{suitable} agent $i^*$ to transfer a good $g \in Z_1 \cap Y_{i^*}$ to $Y_1$. If this transfer creates envy from some agent $j$ (while giving priority to $i^*$) towards $Y_1$ then  agents $1$ and $j$ swap their bundles.

\begin{algorithm}
\LinesNumbered
    \caption{$\eefx + \ef1$ for Restricted Additive Valuations; } \label{alg:eefx+ef1:res-add}
    \DontPrintSemicolon
    \textbf{Input:} A fair division instance $\I = (N, M, V)$ with restricted additive valuations\;
    \textbf{Output:} An $\eefx +\ef1$ allocation\;
    \BlankLine
    \If{$m \leq n$}{
        \Return $(\{g_1\}, \ldots, \{g_m\}, \emptyset, \ldots, \emptyset)$
    }
    Let $X \leftarrow \mathrm{ALG}(I)$ \Comment{\Cref{alg:efxp:res-add}} \; 
    $Y \leftarrow X$ \;
    \While{
    $Z_1 \setminus Y_1 \neq \emptyset$}{
        Let $(i^*,g) \leftarrow \arg \max_{g \in Z_1 \cap Y_{i^*}} u(Y_{i^*} \setminus g)$ \;
        $Y_{i^*} \leftarrow Y_{i^*} \setminus g$ \;
        $Y_1 \leftarrow Y_1 \cup g$ \;
        \If{$\exists j>1$ who envies $Y_1$}{ \label{line}
            \If{$i^*$ envies $Y_1$}{
                $j \leftarrow i^*$ \;
            }
            Swap bundles of agent $1$ and agent $j$ \; \label{line:12}
            \Return $Y$ \;
        }
    }
    \Return $Y$ \;
\end{algorithm}

 Given any specific time during~\Cref{alg:eefx+ef1:res-add}, we assume the allocation is $(Y_1, \ldots, Y_n)$.

\begin{observation}\label{obs:alg-inv}
    Throughout~\Cref{alg:eefx+ef1:res-add}, until the if-condition in Line \eqref{line} is satisfied, the following holds:
    \begin{itemize}
        \item $v_i(Y_i) = u(Y_i)$ for all $i>1$, and 
        \item $v_1(Y_i) = v_1(X_i)$ for all $i \in N$.
    \end{itemize}
\end{observation}

\begin{lemma}\label{lem:ef1-inv}
    Throughout~\Cref{alg:eefx+ef1:res-add}, until Line \eqref{line:12} is executed, the allocation $(Y_2, \ldots, Y_n)$ is $\efx^+$ for the instance where the set of agents is $\{2, \ldots, n\}$ and the set of items is $M \setminus Y_1$.
\end{lemma}
\begin{proof}
    The proof is by induction. By~\Cref{lem:efxp}, the claim holds in the beginning of the algorithm. Assume it holds until the $t$-th iteration of the while loop. Let $(i^*,g) = \arg \max_{g \in Z_1 \cap Y_{i^*}} u(Y_{i^*} \setminus g)$ in the $(t+1)$-st iteration of the while loop. By the condition of the while-loop, $Z_1 \setminus Y_1 \neq \emptyset$. 
    It suffices to prove $i^*$ does not $\efx^+$-envy $Y_j$ for all $j \in N \setminus \{1,i^*\}$. Towards contradiction, assume there exists $j \in N \setminus \{1,i^*\}$ and $h \in Y_j$, such that $v_{i^*}(Y_{i^*}  \setminus g) < v_{i^*} (Y_j \setminus h)$. By~\Cref{obs:alg-inv}, we have
    $$u(Y_{i^*} \setminus g) =  v_{i^*}(Y_{i^*}  \setminus g) < v_{i^*} (Y_j \setminus h) = u(Y_j \setminus h).$$
    Now let $h' = \arg \min_{\ell \in Y_j} u(\ell)$. We have 
    \begin{align}\label{eq1}
        u(Y_{i^*} \setminus g) < u(Y_j \setminus h) \leq u(Y_j \setminus h').
    \end{align}
    By the choice of $(i^*,g)$, $h' \notin Z_1$. By the third point in \Cref{lem:X1}, $Y_j \cap Z_1 = \emptyset$. We get

    \begin{align}
        u(Y_{i^*} \setminus g) &\geq v_1(Y_{i^*} \setminus g) \notag\\
        &= v_1(X_{i^*}) \tag{\Cref{obs:alg-inv}}\\
        &> v_1(X_1) \tag{the first point of \Cref{lem:X1}}\\
        &\geq v_1(X_j \setminus h') \tag{\Cref{lem:efxp}}\\
        &= v_1(Y_j \setminus h') \tag{\Cref{obs:alg-inv}}\\
        &= u(Y_j \setminus h'). \label{eq2}
    \end{align}
    Inequalities \eqref{eq1} and \eqref{eq2} contradict each other. Thus, $i^*$ does not $\efx^+$-envy $Y_j$ for any $j \in N \setminus \{1,i^*\}$.
\end{proof}

\begin{lemma}\label{lemma:eefx-inv}
    During~\Cref{alg:eefx+ef1:res-add}, after moving $g$ from $Y_{i^*}$ to $Y_1$, if $v_{i^*}(Y_{i^*} \setminus g) \geq v_{i^*}(Y_1 \cup g)$, then $Y_{i^*} \setminus g$ is $\eefx$-feasible for $i^*$. If $v_{i^*}(Y_{i^*} \setminus g) < v_{i^*}(Y_1 \cup g)$, then $Y_1 \cup g$ is $\eefx$-feasible for ${i^*}$. Also if $v_j(Y_j) < v_j(Y_1 \cup g)$ for some $j > 1$, then then $Y_1 \cup g$ is $\eefx$-feasible for $j$.
\end{lemma}
\begin{proof}
    By \Cref{lem:ef1-inv}, agent $i^*$ does not $\efx^+$-envy any other agent $j>1$ by having $Y_{i^*} \setminus g$. If $v_{i^*}(Y_{i^*} \setminus g) \geq v_{i^*}(Y_1 \cup g)$, by \Cref{lem:eefx-feas}, $Y_{i^*} \setminus g$ is $\eefx$-feasible for $i^*$.
    Also if $v_{i^*}(Y_{i^*} \setminus g) < v_{i^*}(Y_1 \cup g)$, then by having $Y_1 \cup g$, $i^*$ does not $\efx^+$-envy any other agent and the claim again follows from \Cref{lem:eefx-feas}. By the exact same argument, if $v_j(Y_j) < v_j(Y_1 \cup g)$ for some $j > 1$, then then $Y_1 \cup g$ is $\eefx$-feasible for $j$.
\end{proof}

\thmThree*

\begin{proof}
If $m \leq n$, \Cref{alg:eefx+ef1:res-add} terminates in $O(n)$ time and outputs an EFX (and thus $\efxp$ and EF1) allocation. Now assume $m>n$. \Cref{alg:efxp:res-add} runs in $O(m\log m + nm)$ by \Cref{lem:efxp}. The while-loop can iterate for at most $|Z_1| = O(m)$ many times. Finding $(i^*, g)$ takes $O(m)$ times. Checking the if-conditions takes $O(n)$ time. Hence, overall, the algorithm terminates in $O(m \log m + nm + m^2) = O(m^2)$. 

For analyzing the correctness, we prove the following three claims.

    \begin{claim}
        At the end of the algorithm, agent $1$ receives an $\eefx$-feasible bundle. 
    \end{claim}
    If the condition in Line \ref{line} is never met, the algorithm terminates when $Y_1$ is $\eefx$-feasible for agent $1$. Otherwise, agent $1$ receives some bundle $Y_j$. Note that $X_j \setminus Z_1 \subseteq Y_j$. By \Cref{lem:X1}, $X_j \setminus Z_1$ is $\eefx$-feasible for agent $1$ and therefore $Y_j$ is $\eefx$-feasible too.

    \begin{claim}
        At the end of the algorithm, agent $i$ receives an $\eefx$-feasible bundle for all $i>1$. 
    \end{claim}  
    By \Cref{obs:eefx-feas}, in the beginning of the algorithm, for all $i>1$, $X_i$ is $\eefx$-feasible for $i$. If the condition on Line \eqref{line} is never met, by \Cref{lemma:eefx-inv}, for all $i>1$, $Y_i$ it remains $\eefx$-feasible for its owner until the end of the algorithm. If the condition in Line \eqref{line} is met and ${i^*}$ envies $Y_1 \cup g$, by \Cref{lemma:eefx-inv}, $Y_1 \cup g$ is $\eefx$-feasible for agent ${i^*}$ while the other agents $j>1$ keep their $\eefx$-feasible bundle. If the condition on line \eqref{line} is met and ${i^*}$ does not envy $Y_1 \cup g$ but $j \neq {i^*}$ does, by \Cref{lemma:eefx-inv}, $Y_{i^*} \setminus g$ is $\eefx$-feasible for agent ${i^*}$, $Y_1 \cup g$ is $\eefx$-feasible for $j$, while the other agents $k>1$ keep their bundle $\eefx$-feasible.

    \begin{claim}
        Throughout the algorithm, $Y$ is  EF1.
    \end{claim}
    By~\Cref{lem:efxp}, in the beginning of the algorithm $X$ is $\ef1$. After moving $g$ from $Y_{i^*}$ to $Y_1$, only agent ${i^*}$ can $\ef1$-envy agent $1$, but in that case agent ${i^*}$ receives $Y_1 \cup g$ and agent $1$ receives $Y_{i^*} \setminus g$ and both agents are better off (by the first point of~\Cref{lem:X1}). Also, by~\Cref{lem:ef1-inv}, ${i^*}$ does not $\efx^+$-envy any other bundle and hence does not $\ef1$-envy any other bundle either.    
\end{proof}

\section{Future Directions} \label{sec:conclusion}

We introduced a unifying share-based framework for fair division through the residual maximin share (RMMS), showing that it captures and strengthens several existing lone-divider techniques. Using RMMS, we obtained simple proofs for the existence of RMMS+EFX partial allocations and RMMS+EFL complete allocations. Several questions concerning RMMS remain open.

\begin{enumerate}
    \item Can the RMMS value be computed in pseudopolynomial time? For valuation classes beyond additive valuations, the answer may depend on the available query model.

    \item Can an RMMS-feasible $n$-partition be computed in pseudopolynomial time? This task is not necessarily harder than computing the RMMS value itself. For example, for MXS and additive valuations, an MXS-feasible $n$-partition can be computed in polynomial time by a simple modification of the LMMS algorithm~\cite{lipton2004approximately} combined with the $2/3$-approximation algorithm for MMS~\cite{barman2020approximation}. In contrast, computing the MXS value is weakly NP-hard~\cite{Caragiannis2023}, even for $n=2$.

    \item What is the complexity of computing an RMMS allocation? Observe that the algorithm underlying~\Cref{obs:RSFfeasible} (namely, Steps~1 and~4 of the algorithm used to prove~\Cref{obs:EFXRSF}) runs in polynomial time provided that both the RMMS value and an RMMS-feasible $n$-partition can be computed efficiently. Thus, under sufficiently strong query models, RMMS allocations can be found in polynomial time.
\end{enumerate}

We also demonstrated the versatility of the share-based approach by resolving the compatibility of EEFX and EF1 for additive valuations, establishing the existence of allocations satisfying the stronger combination of EEFX and EFL. Our proof relies on a new share notion, the \emph{strong EEFX share}, which may be of independent interest. Finally, we complemented our existential result with a polynomial-time algorithm for restricted additive valuations.

More broadly, our work highlights the power of share-based fairness notions as a unifying framework for understanding relaxations of EFX and suggests several directions for future research.

\begin{enumerate}
    \setcounter{enumi}{4}

    \item While EEFX and EF1 allocations are known to exist individually for monotone valuations, it remains open whether an allocation satisfying both notions always exists for any valuation class strictly more general than additive valuations.

    \item An analogous question can be posed in the setting of indivisible chores, where agents have non-positive utilities for items.

    \item Our proof of the existence of EEFX+EFL allocations is existential and yields only an exponential-time algorithm. Although EF1 and EEFX allocations can each be computed in polynomial time for additive valuations, the complexity of computing an allocation satisfying both guarantees simultaneously remains open.
\end{enumerate}

\section*{Acknowledgements}
Hannaneh Akrami was supported by a Minerva Fellowship of the Max Planck Society. Uriel Feige was supported in part by the Israel Science Foundation (grant No. 1122/22). Ryoga Mahara was partially supported by the joint project of Kyoto University and Toyota Motor Corporation, titled “Advanced Mathematical Science for Mobility Society”, by JST ERATO Grant Number JPMJER2301, and by JSPS KAKENHI Grant Number JP23K19956, Japan. 

\bibliographystyle{alpha}
\bibliography{ref}
\end{document}

%% file: macros.tex
\newcommand{\efx}{\text{EFX}}
\newcommand{\efxp}{\text{EFX}^+}
\newcommand{\eefx}{\text{EEFX}}
\newcommand{\ef}{\text{EF}}

\newcommand{\mms}{\text{MMS}}

\newcommand{\items}{\mathcal{M}}

\newcommand{\I}{\mathcal{I}}
\newcommand{\instance}{\mathcal{I} = (N,M,\{v_i\}_{i \in N})}
\newcommand{\rmms}{\text{RMMS}}

%% file: new-intro.tex
\section{Introduction}
Fair division studies how to allocate resources among agents in a fair manner. It is a fundamental problem at the intersection of computer science, economics, and social choice theory, with a rich history dating back to the seminal work of Steinhaus 
\cite{S48problem}. Beyond its theoretical significance, fairness considerations are central to the design of many social institutions and arise naturally in diverse allocation settings, with applications ranging from inheritance disputes, course allocation, cloud computing, and public resource allocation \cite{moulin2004fair}. While classical notions such as \emph{envy-freeness}~\cite{foley1966resource} can always be achieved for divisible resources \cite{Su1999rental, stromquist1980cut}, indivisible goods fundamentally change the picture: exact fairness is often impossible, even in very simple instances.\footnote{Consider an instance with two agents and a single item that both agents value positively.} This has led to a rich theory of fairness notions that relax envy-freeness while preserving as much of its appeal as possible.

In this work, we study the classic problem of \emph{fairly} allocating a set of $m$ indivisible items among a set of $n$ agents with \emph{monotone} valuations. Broadly speaking, two complementary approaches have emerged for studying fairness with indivisible goods. \emph{Envy-based} notions evaluate an allocation by comparing an agent’s bundle with those received by other agents. In contrast, \emph{share-based} notions compare an agent’s allocation against a threshold or benchmark. 
Both viewpoints capture natural aspects of fairness, yet neither is entirely satisfactory on its own. An allocation may satisfy an envy-based criterion while providing an agent with far less than what she would consider her fair share, or it may guarantee every agent an appropriate share while leaving substantial envy between agents. Consequently, an important direction in recent years has been to understand the relationship between these two paradigms and to identify allocations that satisfy both envy-based and share-based fairness guarantees simultaneously (see Section~\ref{sec:relate-work}).

\subsection{Share-based Notions and the Residual Maximin Share}
The share-based fairness evaluates an allocation against a threshold defined for each agent. 
The canonical benchmark is the \emph{maximin share (MMS)}, introduced by Budish~\cite{budish2011combinatorial}. Informally, the MMS value of an agent $i$, denoted by $\text{MMS}_i$,  
is the highest value she can guarantee by partitioning the goods into $n$ bundles and receiving the least valuable bundle. Although exact MMS allocations need not exist~\cite{KPW18}, MMS has become one of the central fairness benchmarks for indivisible goods, and a long line of work has established increasingly stronger approximation guarantees (see Section~\ref{sec:relate-work}). 

A common technique used to achieve share-based guarantees is the celebrated \emph{lone-divider} paradigm~\cite{AIGNERHOREV2022164,amanatidis2021maximum,Hummel2025,akrami2025matroidsequitable,EpistemicMonotone}. At a high level, lone-divider algorithms recursively allocate bundles to some agents while ensuring that every remaining agent can still partition the remaining goods into sufficiently valuable bundles. This simple idea has proved remarkably powerful, leading to simultaneous guarantees combining share-based and envy-based fairness \cite{AkramiRathi2025simultaneous}. Despite its success, the structural reason why lone-divider arguments work has remained unclear. This observation motivates the central question underlying our work:
\begin{quote}
  \emph{``What is the intrinsic property of a share notion that makes lone-divider arguments succeed?''}
\end{quote}
The difficulty in answering the above question is that existing share notions do not explicitly capture the invariant maintained by lone-divider algorithms.

Our first conceptual contribution identifies a canonical share notion that directly captures the idea of \emph{residual feasibility} required by lone-divider algorithms, through a new share notion, the \emph{Residual Maximin Share (RMMS)}. Informally, RMMS is the largest share that remains attainable throughout the execution of a lone-divider algorithm: even after previously allocated bundles of smaller value have been removed, every remaining agent can still certify a sufficiently good partition of the remaining goods.

 We show that RMMS satisfies two fundamental properties, even for general monotone valuations. First, it is \emph{feasible}: every fair division instance admits an allocation in which every agent receives at least her RMMS value. Second, it is \emph{self-maximizing}: an agent cannot increase her RMMS value by misreporting her valuation. These properties identify RMMS as the natural share underlying lone-divider arguments, rather than merely another fairness benchmark. In particular, RMMS can be achieved simultaneously with \emph{envy-freeness up to one less-preferred good (EFL)}, and it strictly strengthens several existing guarantees: it dominates the \emph{minimum EFX share (MXS)}~\cite{Caragiannis2023} for all monotone valuations and achieves the $2/3+\Omega(1/n)$-MMS guarantee for additive valuations (see Section~\ref{sec:pre} for the definitions of EFL and MXS).

This perspective yields a unified framework for analyzing lone-divider algorithms. Instead of establishing residual feasibility separately for MMS approximations or other share notions, one can reason directly about RMMS and derive these guarantees as corollaries. 

\color{black}
\subsection{Envy-based Notions}
The notion of \emph{envy-freeness (EF)} is the classical benchmark for fair division \cite{foley1966resource}. An allocation is \emph{envy-free} if every agent weakly prefers her own bundle to that of every other agent. 
Due to the non-existence of EF allocations even in the simplest instances with indivisible items, much of the recent literature has focused on identifying meaningful relaxations that preserve the spirit of envy-freeness while guaranteeing existence.

Among these relaxations, \emph{envy-freeness up to any good (EFX)}, introduced by Caragiannis et al.~\cite{caragiannis2019unreasonable}, has emerged as one of the central notions in the area. Informally, an allocation is EFX if any envy that an agent has towards another disappears after removing \emph{any} single item from the other agent’s bundle. Despite considerable effort over the past several years (see \Cref{sec:relate-work}), whether EFX allocations always exist for additive valuations remains one of the major open problems in fair division. On the other hand, recent impossibility results demonstrate that EFX allocations do not always exist for general monotone (indeed, submodular) valuations~\cite{akrami2026counterexample,mackenzie2026}, making it increasingly important to understand weaker fairness notions that remain achievable in broader domains.

Two relaxations of EFX have received particular attention. The first is envy-freeness up to one good (EF1)~\cite{budish2011combinatorial}, which requires that any envy that an agent has towards another disappears after removing \emph{some} appropriately chosen good from the other bundle. Unlike EFX, EF1 allocations always exist under monotone valuations and can be computed in polynomial time~\cite{lipton2004approximately}. A stronger variant of EF1 is \emph{envy-freeness up to one less-preferred good (EFL)}~\cite{barman2018groupwise}. EFL strengthens EF1 by additionally requiring that the removed good is not more valuable than the envying agent’s own bundle unless there is only one good in the envied bundle that is positively valued by the envious agent. For additive valuations, Barman et al. \cite{barman2018groupwise} established the existence of EFL allocations, and their arguments naturally extend to monotone valuations. 

A conceptually different relaxation of EFX is \emph{epistemic envy-freeness up to any good (EEFX)}, introduced by Caragiannis et al.~\cite{Caragiannis2023}. Informally, a bundle is EEFX-feasible for an agent if the remaining goods can be partitioned so that the resulting allocation would be EFX from that agent’s perspective. An allocation is EEFX if every agent receives an EEFX-feasible bundle. Caragiannis et al.~\cite{Caragiannis2023} proved that EEFX allocations always exist for additive valuations, while Akrami and Rathi \cite{EpistemicMonotone} extended this result to general monotone valuations.

While the existence of allocations satisfying EF1 or EEFX individually is known even for general monotone valuations \cite{lipton2004approximately,EpistemicMonotone}, the existence of an allocation satisfying both notions simultaneously remains open for all settings in which the existence of EFX is itself open. As suggested by Akrami and Rathi \cite{EpistemicMonotone}, the question
\begin{quote}
\emph{``Even for instances with additive valuations, does there always exist an allocation that is simultaneously EEFX and EF1?''}
\end{quote}
is therefore natural and intriguing. This question lies at the intersection of the two principal relaxations of EFX. On the one hand, every EFX allocation is simultaneously EF1 and EEFX. On the other hand, any counterexample to the compatibility of EF1 and EEFX would immediately imply that EFX allocations do not always exist. Hence, resolving this compatibility question provides new structural insight into the elusive EFX problem. We answer the above open compatibility question in \cite{EpistemicMonotone} affirmatively thereby providing new structural insight into the elusive EFX problem.  

\subsection{Technical Contribution}
The paper has two main themes, which we show to be deeply connected. The first is the introduction and study of the \emph{residual maximin share} (RMMS). We show that RMMS satisfies two fundamental properties that one would naturally expect from a share-based fairness notion, even for the broad class of monotone valuations: \emph{feasibility} and \emph{self-maximization}. These properties provide a conceptual justification for RMMS, beyond its role as the natural share underlying the lone-divider arguments. We further show that RMMS strengthens several existing share guarantees: in particular, we show that RMMS dominates the \emph{minimum EFX share (MXS)} for monotone valuations and achieves the $2/3+\Omega(1/n)$-MMS guarantee for additive valuations. In addition, we prove that RMMS is compatible with \emph{envy-freeness up to one less-preferred good (EFL)}: every instance admits an allocation satisfying both RMMS and EFL.

The second theme concerns the compatibility of two of the most prominent relaxations of EFX: EFL (and its weaker variant EF1) and EEFX. We show that these guarantees can be achieved simultaneously for additive valuations. To this end, we introduce a new share-based notion that implies EEFX-feasibility and show that it is upper bounded by RMMS. Consequently, every simultaneous guarantee established for RMMS also extends to EEFX. More broadly, this result establishes a connection between fairness notions that have largely been studied independently and yields, to the best of our knowledge, the strongest relaxation of EFX currently known.

Finally, concerning the second theme, we develop a polynomial-time algorithm for computing allocations that simultaneously satisfy EEFX and EF1 under restricted additive valuations. In contrast to our existence results, this algorithm does not rely on the lone-divider framework. Instead, it exploits the structural properties of restricted additive valuations to efficiently construct an EEFX+EF1 allocation.
\color{black}
\paragraph{Residual Maximin Share.} In this work, we first formalize the argument  of Akrami and Rathi \cite{AkramiRathi2025simultaneous} (described in Section~\ref{sec:RMMS}) more generally and provide necessary and sufficient conditions for the lone-divider algorithm to work. In particular, for a given fair division instance, we define the \emph{residual maximin share} $\rmms_i$ as the maximum value for which agent $i$ can act as the lone divider (see~\Cref{sec:pre} for a formal definition). We obtain $\text{MXS}_i\leq\rmms_i \leq \text{MMS}_i$, where the first inequality is discussed in~\Cref{obs:MXS}, while the second inequality follows by definition. An allocation is RMMS if all agents value their bundle at least as much as their RMMS value. Building on this, the extension suggested by Akrami and Rathi \cite{AkramiRathi2025simultaneous} proves the existence of partial allocations that are EFX and RMMS and complete allocations that are EF1 and RMMS. Moreover, we observe that EF1 can be strengthened to EFL, yielding a slightly stronger guarantee.

\begin{restatable}{proposition}{propFeige}\label{prop:Feige}
    For every fair division instance with monotone valuations,
    \begin{enumerate}
        \item there exists a partial allocation that is simultaneously RMMS and EFX, and
        \item there exists a complete allocation that is simultaneously RMMS and EFL (and hence EF1).
    \end{enumerate}    
\end{restatable}

This proposition unifies and simplifies several previous lone-divider analyses. In particular, it shows that the correctness of these algorithms depends only on residual self-feasibility, rather than on the specific share notion being used.

In \Cref{sec:RMMS}, we establish several useful properties of RMMS, showing that it is a well-motivated share-based notion in its own right. We further place RMMS in the broader landscape of share guarantees by comparing it to several prominent notions from the literature like MMS and its approximations and MXS.

\paragraph{Existence of EEFX+EFL Allocations for Additive Valuations.} 
 Our second contribution demonstrates the power of the RMMS framework by applying it to the compatibility of fairness notions related to the long-standing EFX problem. Although both EEFX and EF1 are individually guaranteed to exist, whether they can always be achieved simultaneously has remained open in every setting where the existence of EFX itself is unresolved. 

 At first glance, a direct application of the RMMS framework technique might not seem helpful in proving EF1 and EEFX allocations, since EEFX is not a share-based notion. 
It seems not so clear how exactly  one should categorize EEFX. 
Unlike standard envy-based notions, EEFX feasibility of a bundle is independent of \emph{how} the remaining items are allocated, a feature typically associated with share-based concepts. On the other hand, unlike classical share-based notions, EEFX feasibility is not determined solely by the \emph{value} of a bundle. In particular, there may exist bundles $A$ and $B$ such that agent $i$ values $A$ more than $B$, yet $B$ is EEFX feasible for $i$ while $A$ is not (see \Cref{example}). This shows that while MXS can be guaranteed by the lone divider approach and is compatible with EF1, it does not imply that the final allocation is EEFX.

To overcome this obstacle, we introduce another share notion, the strong EEFX share, which captures the threshold beyond which EEFX feasibility becomes monotone. Our main structural theorem proves that for additive valuations, the strong EEFX share is always upper bounded by RMMS. This bridge allows us to leverage the RMMS framework to establish the following existence theorem.

\begin{restatable}{theorem}{thmMain}\label{thm:main}
    Every fair division instance with additive valuations admits an allocation that is simultaneously EEFX and EFL.
\end{restatable}

Since every EFL allocation is also EF1, this immediately implies the existence of allocations satisfying both EEFX and EF1, thereby resolving the open problem of Akrami and Rathi \cite{EpistemicMonotone}. We emphasize that our proof is existential and naturally translates into an exponential-time algorithm. Designing a polynomial-time algorithm for computing EEFX+EFL allocations remains open for additive valuations. 

To prove the above theorem, we introduce a new share-based notion, termed as \emph{strong EEFX share}, which may be of independent interest. While EEFX itself is not a share-based notion, the strong EEFX share captures the threshold beyond which EEFX feasibility becomes \emph{monotone}. For an agent~$i$, we define her \emph{strong EEFX share}, denoted by $\theta_i$, as the minimum threshold value such that every bundle with value at least this threshold is EEFX feasible for $i$. Since MMS implies EEFX feasibility \cite{Caragiannis2023}, the strong EEFX share is upper-bounded by MMS value. By definition of MXS, we also know that strong EEFX share is lower-bounded by MXS value. Crucially, we show that the strong EEFX share has the necessary conditions for the lone divider approach to work, i.e., the strong EEFX share is at most the RMMS value for an agent (see~\Cref{thm:self-feasible}). 
Overall, for an agent $i$, we establish the following relation between her strong EEFX share and other share-based fair guarantees, 
$$\text{MXS}_i \leq \theta_i \leq \rmms_i \leq \text{MMS}_i.$$

Figure~\ref{fig:0} depicts the relations between RMMS and strong EEFX share with other known fairness notions.

\begin{figure}[t]
    \centering
\includegraphics[width=0.90\linewidth]{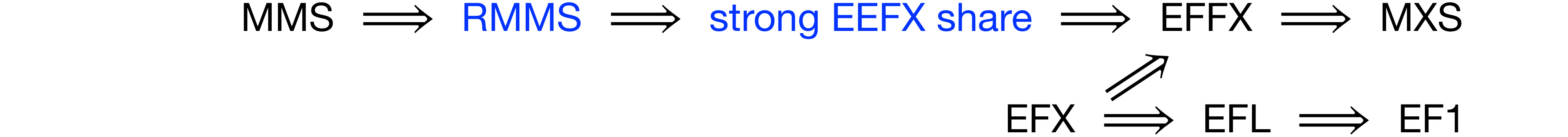}
    \caption{Relationships among fairness notions under additive valuations.}
    \label{fig:0}
\end{figure}

Using~\Cref{thm:self-feasible}, we obtain the existence of a partial EFX allocation, which can be extended to a complete EFL allocation, wherein every agent receives a bundle whose value is at least their strong EEFX share. By definition, this implies that the allocation is EEFX, thereby establishing the existence of allocations that are simultaneously EEFX and EFL (and hence, EF1) for additive valuations. 

Furthermore, \Cref{example} shows that the strong EEFX share is strictly stronger than the MXS share, thereby strengthening the result of \cite{Ashuri2025}, who established the existence of EFL allocations guaranteeing every agent their MXS share.

\paragraph{Efficient Computation for Restricted Additive Valuations.} A natural question is whether allocations satisfying both EEFX and EFL can be computed in polynomial time. While our RMMS-based framework establishes existence, it does not yield a polynomial-time algorithm. Our third contribution complements this existential result with an efficient algorithm for an important special class of valuations - \emph{restricted} additive valuations.

\begin{restatable}{theorem}{thmThree}\label{thm:res}
    There exists an $O(m^2)$-time algorithm that computes an allocation satisfying both EEFX and EF1 for restricted additive valuations.
\end{restatable}

Interestingly, this algorithm is conceptually unrelated to our existential proof. Rather than relying on the lone-divider paradigm, it exploits structural properties specific to restricted additive valuations to iteratively modify an initial allocation while preserving EF1 and progressively establishing EEFX feasibility. 

\paragraph{Novelty and significance.} Taken together, our results illustrate the broader usefulness of share-based reasoning in understanding envy-based fairness. 
At a conceptual level, the notion of RMMS identifies the residual feasibility invariant underlying lone-divider algorithms and provides a reusable framework for analyzing recursive fair allocation techniques. At the technical level, our framework resolves a natural open problem and proves the existence of EEFX+EF1 allocations for additive valuations and develops a polynomial-time algorithm for restricted additive valuations. 
We believe that both RMMS and the strong EEFX share may prove useful beyond the specific applications considered in this paper.

\subsection{Further Related Work}\label{sec:relate-work}
After \cite{KPW18} showed that MMS allocations are not guaranteed to exist, a sequence of works established the existence of $\alpha$-MMS allocations for additive valuations \cite{amanatidis2017approximation,ghodsi2018fair,barman2020approximation,garg2020improved,FST21,simple,akrami2024breaking,heidari2025improved1013}, culminating in the current state of the art of $\alpha = 7/9$ \cite{huang2025fptas79}.

Beyond MMS, several share-based fairness notions have been studied in discrete fair division. Proportionality up to one item (PROP1) \cite{conitzer2017fair} can be guaranteed simultaneously with Pareto optimality \citep{barman2019proximity}, whereas proportionality up to any item (PROPx) is known not to be guaranteed to exist in the goods setting \cite{aziz2020polynomial}.

The existence of EFX allocations is known only in very restricted settings; e.g., for a limited number of agents \cite{plaut2020almost,chaudhury2020efx,berger2021almost,AACGMM25}, limited number of items \cite{amanatidis2020multiple,mahara2024extension}, and for special valuation functions \cite{plaut2020almost,halpern2020fair,amanatidis2021maximum}. Very recently, \cite{akrami2026counterexample} constructed a counterexample demonstrating the non-existence of EFX allocations for submodular valuations when $n \geq 3$ agents and $m \geq n+5$ goods, by formulating the problem as a satisfiability instance. Subsequently, \cite{mackenzie2026} gave a human-verifiable counterexample in the same regime.

Analogous to $\alpha$-MMS, one can define $\alpha$-EFX, which requires that
$v_i(X_i) \ge \alpha \cdot v_i(X_j \setminus {g})$ for all agents $i,j$ and all items $g \in X_j$. The existence of $\alpha$-EFX allocations has been established in a number of works \cite{chaudhury2021little,amanatidis2020multiple,chan2019maximin,farhadi2021almost}. While the existence of EFX allocations for four agents with additive valuations remains open, \cite{VishwaEFX} proved the existence of EFX allocations for three classes of valuations. Furthermore, EF2X—a relaxation of EFX requiring that the removal of any two items from an envied bundle eliminates any envy—is guaranteed to exist for restricted additive valuations \cite{ef2x-restricted} and for four agents with cancelable valuations, a class that strictly generalizes additive valuations \cite{ef2x-4-agents}. 

Initiated by Amanatidis et al. \cite{amanatidis2020multiple}, subsequent work has studied the compatibility of various share-based notions with different envy-based notions \cite{chaudhury2021little,AkramiRathi2025simultaneous}.

For a comprehensive discussion of these results and related literature, we refer the reader to the survey by~\cite{survey2022}. A closely related and extensively studied line of work concerns the allocation of \emph{chores} rather than goods; see~\cite{guo2023survey} for a comprehensive overview.

\subsection{Organization}
In Section~\ref{sec:pre}, we introduce the fair division model, relevant notions of fairness, and introduce key new concepts of residual self-feasibility, self-maximization, RMMS, and strong EEFX share. Section~\ref{sec:lone} describes the lone-divider approach. Section~\ref{sec:RMMS} proves the important properties of RMMS for monotone valuations, establishes Proposition~\ref{prop:Feige}, and discusses the implications of RMMS to other fairness notions. It then discusses the computational aspects of RMMS. Section~\ref{sec:eefx-efl} shows the existence of allocations that are simultaneously EEFX and EFL for additive valuations (Theorem~\ref{thm:main}). Section~\ref{sec:res-add} develops a polynomial-time algorithm for finding EEFX+EF1 allocations for restricted additive valuations (Theorem~\ref{thm:res}). Finally, Section~\ref{sec:conclusion} discusses the conclusions and directions for future work.

%% file: ref.bib
@article{F25,
  title={The residual maximin share},
  author={Feige, Uriel},
  journal={arXiv preprint arXiv:2505.19961},
  year={2025}
}

@article{ABM18,
  title={Comparing approximate relaxations of envy-freeness},
  author={Amanatidis, Georgios and Birmpas, Georgios and Markakis, Evangelos},
  journal={arXiv preprint arXiv:1806.03114},
  year={2018}
}

@article{barman2020approximation,
  title={Approximation algorithms for maximin fair division},
  author={Barman, Siddharth and Krishnamurthy, Sanath Kumar},
  journal={ACM Transactions on Economics and Computation (TEAC)},
  volume={8},
  number={1},
  pages={1--28},
  year={2020},
  publisher={ACM New York, NY, USA}
}

@article{bu2024fair,
  title={Fair division of indivisible goods with comparison-based queries},
  author={Bu, Xiaolin and Li, Zihao and Liu, Shengxin and Song, Jiaxin and Tao, Biaoshuai},
  journal={arXiv e-prints},
  pages={arXiv--2404},
  year={2024}
}

@article{BUF23,
  title={On fair allocation of indivisible goods to submodular agents},
  author={Uziahu, Gilad Ben and Feige, Uriel},
  journal={arXiv preprint arXiv:2303.12444},
  year={2023}
}

@article{BF25,
  title={Fair shares: Feasibility, domination, and incentives},
  author={Babaioff, Moshe and Feige, Uriel},
  journal={Mathematics of Operations Research},
  volume={50},
  number={3},
  pages={1901--1934},
  year={2025},
  publisher={INFORMS}
}

@article{plaut2020almost,
  title={Almost envy-freeness with general valuations},
  author={Plaut, Benjamin and Roughgarden, Tim},
  journal={SIAM Journal on Discrete Mathematics},
  volume={34},
  number={2},
  pages={1039--1068},
  year={2020},
  publisher={SIAM}
}

@article{survey2022,
title = {Fair division of indivisible goods: Recent progress and open questions},
journal = {Artificial Intelligence},
volume = {322},
pages = {103965},
year = {2023},
author = {Georgios Amanatidis and Haris Aziz and Georgios Birmpas and Aris Filos-Ratsikas and Bo Li and Hervé Moulin and Alexandros A. Voudouris and Xiaowei Wu}
}

@article{budish2011combinatorial,
  title={The Combinatorial Assignment Problem: Approximate Competitive Equilibrium from Equal Incomes},
  author={Eric Budish},
  journal={Journal of Political Economy},
  volume={119},
  number={6},
  pages={1061--1103},
  year={2011},
  publisher={University of Chicago Press}
}

@book{foley1966resource,
  title={Resource allocation and the public sector},
  author={Duncan Karl Foley},
  year={1966}, 
  volume={7:45-98},
  publisher={Yale Economic Essays},
  address = {New Haven, CT, USA}
}

@article{guo2023survey,
  title={A survey on fair allocation of chores},
  author={Hao Guo and Weidong Li and Bin Deng},
  journal={Mathematics},
  volume={11},
  number={16},
  pages={3616},
  year={2023},
  publisher={MDPI}
}

@inproceedings{halpern2020fair,
  author       = {Daniel Halpern and
                  Ariel D. Procaccia and
                  Alexandros Psomas and
                  Nisarg Shah},
  title        = {Fair Division with Binary Valuations: One Rule to Rule Them All},
  booktitle    = {Web and Internet Economics - 16th International Conference, {WINE}},
  series       = {Lecture Notes in Computer Science},
  volume       = {12495},
  pages        = {370--383},
  publisher    = {Springer},
year = {2020},
address = {Beijing, China}
}

@inproceedings{lipton2004approximately,
  title={On approximately fair allocations of indivisible goods},
  author={Richard J. Lipton and
                  Evangelos Markakis and
                  Elchanan Mossel and
                  Amin Saberi},
  publisher = {{ACM}},
  booktitle = {Proceedings of the 5th {ACM} Conference on Electronic Commerce ({EC})},
  pages={125--131},
  year={2004},
  address = {New York, NY, USA}
}

@book{moulin2004fair,
  author       = {Herv{\'{e}} Moulin},
  title        = {Fair division and collective welfare},
  publisher    = {{MIT} Press},
  year         = {2003},
  address = {Cambridge, MA, USA}
}

@article{stromquist1980cut,
  title={How to Cut a Cake Fairly},
  author={Walter Stromquist},
  journal={The American Mathematical Monthly},
  volume={87},
  number={8},
  pages={640--644},
  year={1980}
}

@article{Su1999rental,
  title={Rental Harmony: Sperner's Lemma in Fair Division},
  author={Francis Edward Su},
  journal={The American Mathematical Monthly},
  volume={106},
  number={10},
  pages={930--942},
  year={1999}
}

@inproceedings{Caragiannis2023,
  title     = {New Fairness Concepts for Allocating Indivisible Items},
  author    = {Caragiannis, Ioannis and Garg, Jugal and Rathi, Nidhi and Sharma, Eklavya and Varricchio, Giovanna},
  booktitle = {Proceedings of the Thirty-Second International Joint Conference on
               Artificial Intelligence, {IJCAI-23}},
  publisher = {International Joint Conferences on Artificial Intelligence Organization},
  pages     = {2554--2562},
  year      = {2023}
}

@inproceedings{EpistemicMonotone,
  author       = {Hannaneh Akrami and
                  Nidhi Rathi},
  title        = {Epistemic {EFX} Allocations Exist for Monotone Valuations},
  booktitle    = {AAAI-25, Sponsored by the Association for the Advancement of Artificial
                  Intelligence},
  pages        = {13520--13528},
  publisher    = {{AAAI} Press},
  year         = {2025}
}

@article{chaudhury2024efx,
  title={{EFX} exists for three agents},
  author={Chaudhury, Bhaskar Ray and Garg, Jugal and Mehlhorn, Kurt},
  journal={Journal of the ACM},
  volume={71},
  number={1},
  pages={1--27},
  year={2024},
  publisher={ACM New York, NY}
}

@inproceedings{Ashuri2025,
  author       = {Arash Ashuri and
                  Vasilis Gkatzelis},
  editor       = {Itai Ashlagi and
                  Aaron Roth},
  title        = {Simultaneously Satisfying {MXS} and {EFL}},
  booktitle    = {Proceedings of the 26th {ACM} Conference on Economics and Computation,
                  {EC} 2025, Stanford University, Stanford, CA, USA, July 7-10, 2025},
  pages        = {689--718},
  publisher    = {{ACM}},
  year         = {2025},
  url          = {https://doi.org/10.1145/3736252.3742613}
}

@inproceedings{AkramiRathi2025simultaneous,
author = {Akrami, Hannaneh and Rathi, Nidhi},
title = {Achieving maximin share and {EFX}/{EF1} guarantees simultaneously},
year = {2025},
isbn = {978-1-57735-897-8},
publisher = {AAAI Press},
url = {https://doi.org/10.1609/aaai.v39i13.33477},
doi = {10.1609/aaai.v39i13.33477},
booktitle = {Proceedings of the Thirty-Ninth AAAI Conference on Artificial Intelligence and Thirty-Seventh Conference on Innovative Applications of Artificial Intelligence and Fifteenth Symposium on Educational Advances in Artificial Intelligence},
articleno = {1504},
numpages = {9},
}

@article{mahara2024extension,
  title={Extension of additive valuations to general valuations on the existence of {EFX}},
  author={Mahara, Ryoga},
  journal={Mathematics of operations research},
  volume={49},
  number={2},
  pages={1263--1277},
  year={2024},
  publisher={INFORMS}
}

@article{S48problem,
  title={{The Problem of Fair Division}},
  author={Steinhaus, Hugo},
  journal={Econometrica},
  volume={16},
  pages={101--104},
  year={1948}
}

@article{AACGMM25,
  author       = {Hannaneh Akrami and
                  Noga Alon and
                  Bhaskar Ray Chaudhury and
                  Jugal Garg and
                  Kurt Mehlhorn and
                  Ruta Mehta},
  title        = {{EFX:} {A} Simpler Approach and an (Almost) Optimal Guarantee via
                  Rainbow Cycle Number},
  journal      = {Oper. Res.},
  volume       = {73},
  number       = {2},
  pages        = {738--751},
  year         = {2025},
  url          = {https://doi.org/10.1287/opre.2023.0433},
  doi          = {10.1287/OPRE.2023.0433},
  bibsource    = {dblp computer science bibliography, https://dblp.org}
}

@inproceedings{berger2021almost,
  title={Almost Full {EFX} Exists for Four Agents},
  author={Berger, Ben and Cohen, Avi and Feldman, Michal and Fiat, Amos},
  booktitle={Proceedings of the 36th AAAI Conference on Artificial Intelligence ({AAAI})},
  volume={36(5)},
  year={2022},
  pages={4826-4833}
}

@article{amanatidis2020multiple,
  title={Multiple birds with one stone: Beating 1/2 for {EFX} and {GMMS} via envy cycle elimination},
  author={Amanatidis, Georgios and Markakis, Evangelos and Ntokos, Apostolos},
  journal={Theoretical Computer Science},
  volume={841},
  pages={94--109},
  year={2020},
  publisher={Elsevier}
}

@article{amanatidis2021maximum,
  title={Maximum {Nash} welfare and other Stories about {EFX}},
  author={Amanatidis, Georgios and Birmpas, Georgios and Filos-Ratsikas, Aris and Hollender, Alexandros and Voudouris, Alexandros A},
  journal={Theoretical Computer Science},
  volume={863},
  pages={69--85},
  year={2021},
  publisher={Elsevier}
}

@article{AIGNERHOREV2022164,
title = {Envy-free matchings in bipartite graphs and their applications to fair division},
journal = {Information Sciences},
volume = {587},
pages = {164-187},
year = {2022},
issn = {0020-0255},
doi = {https://doi.org/10.1016/j.ins.2021.11.059},
url = {https://www.sciencedirect.com/science/article/pii/S0020025521011816},
author = {Elad Aigner-Horev and Erel Segal-Halevi},
}

@article{Hummel2025,
author = {Hummel, Halvard},
title = {Maximin Shares in Hereditary Set Systems},
year = {2025},
issue_date = {September 2025},
publisher = {Association for Computing Machinery},
address = {New York, NY, USA},
volume = {13},
number = {3},
issn = {2167-8375},
url = {https://doi.org/10.1145/3727149},
doi = {10.1145/3727149},
journal = {ACM Trans. Econ. Comput.},
month = jun,
articleno = {12},
numpages = {33}
}

@article{chaudhury2021little,
  title={A Little Charity Guarantees Almost Envy-Freeness},
  author={Chaudhury, Bhaskar Ray and Telikepalli, Kavitha and Mehlhorn, Kurt and Sgouritsa, Alkmini},
  journal={SIAM Journal on Computing},
  volume={50},
  number={4},
  pages={1336--1358},
  year={2021}
}

@article{KPW18,
  title={Fair enough: Guaranteeing approximate maximin shares},
  author={Kurokawa, David and Procaccia, Ariel D. and Wang, Junxing},
  journal={Journal of the ACM},
  volume={65},
  number={2},
  pages={1--27},
  year={2018},
  publisher={ACM New York, NY, USA}
}

@article{caragiannis2019unreasonable,
  title={The unreasonable fairness of maximum {N}ash welfare},
  author={Caragiannis, Ioannis and Kurokawa, David and Moulin, Herv{\'e} and Procaccia, Ariel D and Shah, Nisarg and Wang, Junxing},
  journal={ACM Transactions on Economics and Computation (TEAC)},
  volume={7},
  number={3},
  pages={1--32},
  year={2019},
  publisher={ACM New York, NY, USA}
}

@inproceedings{barman2018groupwise,
  title={Groupwise maximin fair allocation of indivisible goods},
  author={Barman, Siddharth and Biswas, Arpita and Krishnamurthy, Sanath and Narahari, Yadati},
  booktitle={Proceedings of the AAAI Conference on Artificial Intelligence},
  volume={32},
  number={1},
  year={2018}
}

@inproceedings{chaudhury2020efx,
  title={{EFX} Exists for Three Agents},
  author={Chaudhury, Bhaskar Ray and Garg, Jugal and Mehlhorn, Kurt},
  publisher = {{ACM}},
  booktitle = {Proceedings of the 21st ACM Conference on Economics and Computation ({EC})},
  pages={1--19},
  year={2020}
}

@article{amanatidis2017approximation,
  title={Approximation algorithms for computing maximin share allocations},
  author={Amanatidis, Georgios and Markakis, Evangelos and Nikzad, Afshin and Saberi, Amin},
  journal={ACM Transactions on Algorithms},
  volume={13},
  number={4},
  pages={1--28},
  year={2017},
  publisher={ACM New York, NY, USA}
}

@inproceedings{ghodsi2018fair,
  title={Fair allocation of indivisible goods: Improvements and generalizations},
  author={Ghodsi, Mohammad and Hajiaghayi, MohammadTaghi and Seddighin, Masoud and Seddighin, Saeed and Yami, Hadi},
  booktitle={Proceedings of the 19th ACM Conference on Economics and Computation ({EC})},
  pages={539--556},
  year={2018}
}

@inproceedings{garg2020improved,
  title={An improved approximation algorithm for maximin shares},
  author={Garg, Jugal and Taki, Setareh},
  booktitle={Proceedings of the 21st ACM Conference on Economics and Computation ({EC})},
  pages={379--380},
  year={2020}
}

@inproceedings{FST21,
  title={A tight negative example for {MMS} fair allocations},
  author={Feige, Uriel and Sapir, Ariel and Tauber, Laliv},
  booktitle={Proceedings of the 17th International Conference on Web and Internet Economics {(WINE)}},
  pages={355--372},
  year={2021},
  organization={Springer}
}

@inproceedings{simple,
  title={Simplification and Improvement of {MMS} Approximation},
  author={Hannaneh Akrami and Jugal Garg and Eklavya Sharma and Setareh Taki},
  booktitle={Proceedings of the 32nd International Joint Conference on Artificial Intelligence (IJCAI)},
  year={2023}
}

@inproceedings{akrami2024breaking,
  title={Breaking the 3/4 barrier for approximate maximin share},
  author={Akrami, Hannaneh and Garg, Jugal},
  booktitle={Proceedings of the 2024 Annual ACM-SIAM Symposium on Discrete Algorithms (SODA)},
  pages={74--91},
  year={2024},
  organization={SIAM}
}

@article{heidari2025improved1013,
  title={Improved maximin share guarantee for additive valuations},
  author={Heidari, Ehsan and Kaviani, Alireza and Seddighin, Masoud and Shahrezaei, AmirMohammad},
  journal={arXiv preprint arXiv:2510.10423},
  year={2025}
}

@article{huang2025fptas79,
  title={An FPTAS for 7/9-Approximation to Maximin Share Allocations},
  author={Huang, Xin and Zhou, Shengwei},
  journal={arXiv preprint arXiv:2511.13056},
  year={2025}
}

@inproceedings{VishwaEFX,
author = {Hv, Vishwa Prakash and Ghosal, Pratik and Nimbhorkar, Prajakta and Varma, Nithin},
title = {EFX Exists for Three Types of Agents},
year = {2025},
isbn = {9798400719431},
publisher = {Association for Computing Machinery},
address = {New York, NY, USA},
url = {https://doi.org/10.1145/3736252.3742509},
doi = {10.1145/3736252.3742509},
booktitle = {Proceedings of the 26th ACM Conference on Economics and Computation},
pages = {101–128},
numpages = {28},
location = {Stanford University, Stanford, CA, USA},
series = {EC '25}
}

@inproceedings{chan2019maximin,
  title={Maximin-aware allocations of indivisible goods},
  author={Chan, Hau and Chen, Jing and Li, Bo and Wu, Xiaowei},
  booktitle={Proceedings of the 28th International Joint Conference on Artificial Intelligence ({IJCAI})},
  pages={137--143},
  year={2019}
}

@inproceedings{farhadi2021almost,
  title={Almost envy-freeness, envy-rank, and Nash social welfare matchings},
  author={Farhadi, Alireza and Hajiaghayi, MohammadTaghi and Latifian, Mohamad and Seddighin, Masoud and Yami, Hadi},
  booktitle={Proceedings of the AAAI Conference on Artificial Intelligence},
  volume={35},
  number={6},
  pages={5355--5362},
  year={2021}
}

@inproceedings{ef2x-4-agents,
  author       = {Arash Ashuri and
                  Vasilis Gkatzelis and
                  Alkmini Sgouritsa},
  editor       = {Toby Walsh and
                  Julie Shah and
                  Zico Kolter},
  title        = {{EF2X} Exists for Four Agents},
  booktitle    = {AAAI-25, Sponsored by the Association for the Advancement of Artificial
                  Intelligence, February 25 - March 4, 2025, Philadelphia, PA, {USA}},
  pages        = {13555--13563},
  publisher    = {{AAAI} Press},
  year         = {2025}
}

@inproceedings{ef2x-restricted,
  author       = {Hannaneh Akrami and
                  Rojin Rezvan and
                  Masoud Seddighin},
  editor       = {Luc De Raedt},
  title        = {An {EF2X} Allocation Protocol for Restricted Additive Valuations},
  booktitle    = {Proceedings of the Thirty-First International Joint Conference on
                  Artificial Intelligence, {IJCAI} 2022, Vienna, Austria, 23-29 July
                  2022},
  pages        = {17--23},
  publisher    = {ijcai.org},
  year         = {2022}
}

@inproceedings{conitzer2017fair,
  title={Fair public decision making},
  author={Conitzer, Vincent and Freeman, Rupert and Shah, Nisarg},
  booktitle={Proceedings of the 2017 ACM Conference on Economics and Computation},
  pages={629--646},
  year={2017}
}

@article{aziz2020polynomial,
  title={A polynomial-time algorithm for computing a {P}areto optimal and almost proportional allocation},
  author={Aziz, Haris and Moulin, Herv{\'e} and Sandomirskiy, Fedor},
  journal={Operations Research Letters},
  volume={48},
  number={5},
  pages={573--578},
  year={2020},
  publisher={Elsevier},
  doi={10.1016/j.orl.2020.07.005}
}

@inproceedings{barman2019proximity,
  title={On the proximity of markets with integral equilibria},
  author={Barman, Siddharth and Krishnamurthy, Sanath Kumar},
  booktitle={Proceedings of the AAAI Conference on Artificial Intelligence},
  volume={33},
  number={01},
  pages={1748--1755},
  year={2019}
}

@misc{akrami2025matroidsequitable,
      title={Matroids are Equitable}, 
      author={Hannaneh Akrami and Siyue Liu and Roshan Raj and László A. Végh},
      year={2025},
      eprint={2507.12100},
      archivePrefix={arXiv},
      primaryClass={math.CO},
      url={https://arxiv.org/abs/2507.12100}, 
}

@misc{akrami2026counterexample,
      title={A Counterexample to EFX $n \ge 3$ Agents, $m \ge n + 5$ Items, Submodular Valuations via SAT-Solving}, 
      author={Hannaneh Akrami and Alexander Mayorov and Kurt Mehlhorn and Shreyas Srinivas and Christoph Weidenbach},
      year={2026},
      eprint={2604.18216},
      archivePrefix={arXiv},
      primaryClass={cs.GT},
      url={https://arxiv.org/abs/2604.18216}, 
}

@misc{mackenzie2026,
      title={Counterexamples to EFX for Submodular and Subadditive Valuations}, 
      author={Simon Mackenzie and Mashbat Suzuki},
      year={2026},
      eprint={2605.06451},
      archivePrefix={arXiv},
      primaryClass={cs.GT},
      url={https://arxiv.org/abs/2605.06451}, 
}

@article{GhodsiHSSY22,
title = {Fair allocation of indivisible goods: Beyond additive valuations},
journal = {Artificial Intelligence},
volume = {303},
pages = {103633},
year = {2022},
issn = {0004-3702},
doi = {https://doi.org/10.1016/j.artint.2021.103633},
url = {https://www.sciencedirect.com/science/article/pii/S0004370221001843},
author = {Mohammad Ghodsi and MohammadTaghi HajiAghayi and Masoud Seddighin and Saeed Seddighin and Hadi Yami}
}

@misc{CCMS25,
      title={Maximin Share Guarantees for Few Agents with Subadditive Valuations}, 
      author={George Christodoulou and Vasilis Christoforidis and Symeon Mastrakoulis and Alkmini Sgouritsa},
      year={2025},
      eprint={2502.05141},
      archivePrefix={arXiv},
      primaryClass={cs.GT},
      url={https://arxiv.org/abs/2502.05141}, 
}

@misc{FH25,
      title={Concentration and maximin fair allocations for subadditive valuations}, 
      author={Uriel Feige and Shengyu Huang},
      year={2025},
      eprint={2502.13541},
      archivePrefix={arXiv},
      primaryClass={cs.GT},
      url={https://arxiv.org/abs/2502.13541}, 
}

@inproceedings{SS25,
author = {Seddighin, Masoud and Seddighin, Saeed},
title = {Beating the Logarithmic Barrier for the Subadditive Maximin Share Problem},
year = {2025},
isbn = {9798400719431},
publisher = {Association for Computing Machinery},
address = {New York, NY, USA},
url = {https://doi.org/10.1145/3736252.3742621},
doi = {10.1145/3736252.3742621},
booktitle = {Proceedings of the 26th ACM Conference on Economics and Computation},
pages = {764–782},
numpages = {19},
location = {Stanford University, Stanford, CA, USA},
series = {EC '25}
}

@misc{feige25,
      title={From multi-allocations to allocations, with subadditive valuations}, 
      author={Uriel Feige},
      year={2025},
      eprint={2506.21493},
      archivePrefix={arXiv},
      primaryClass={cs.GT},
      url={https://arxiv.org/abs/2506.21493}, 
}
